\documentclass{llncs}
\usepackage{microtype}
\usepackage{amsmath}
\usepackage{amsfonts}
\usepackage{amssymb}
\usepackage{mathrsfs}
\usepackage{enumerate}
\usepackage{latexsym,graphicx}
\usepackage{xypic}

\xyoption{all}

\def\carac#1,#2{
\left[
\begin{smallmatrix}
#1 \\ #2
\end{smallmatrix}
\right]
}

\newtheorem{algorithm}[theorem]{Algorithm}

\newcommand{\Aff}{\mathbb A} 
 
\newcommand{\N}{\mathbb N}
\newcommand{\Z}{\mathbb{Z}}
\newcommand{\xz}{\mathbb Z}

\newcommand{\proj}{\mathbb P}
\newcommand{\F}{\mathbb F}

\newcommand{\abb}[5]{#1:#2\rightarrow#3,#4\mapsto#5}

\newcommand{\hsp}{\hspace{5pt}}
\newcommand{\spec}[1]{ \mathrm{Spec} (#1)}

\newcommand{\xg}{\mathbb{G}}

\newcommand{\inzweizweimat}[4]{\begin{array}{c@{\quad}c}
#1 & #2 \\
#3 & #4
\end{array}}
 
\newcommand{\squaremat}[4]{\left( \inzweizweimat{#1}{#2}{#3}{#4} \right)}

\newcommand{\xpol}{\mathscr{X}}
\newcommand{\pol}{\mathscr{L}}
\newcommand{\ppol}{\mathscr{L}_0}

\newcommand{\bpol}{\mathscr{M}_0}
\newcommand{\pppol}{\mathscr{M}}

\renewcommand*{\emptyset}{\varnothing}
\renewcommand*{\phi}{\varphi}
\renewcommand*{\epsilon}{\varepsilon}

\renewcommand{\leq}{\ensuremath{\leqslant}}
\renewcommand{\geq}{\ensuremath{\geqslant}}

\DeclareMathOperator{\Aut}{Aut}
\DeclareMathOperator{\Auts}{Aut_s}

\DeclareMathOperator{\Id}{Id}

\DeclareMathOperator{\SProj}{Proj}

\newcommand{\kbar}{\overline{k}}

\newcommand{\overln}{\ensuremath{\overline{\ell n}}}
\newcommand{\overn}{\ensuremath{\overline{n}}}
\newcommand{\overl}{\ensuremath{\overline{\ell}}}
\newcommand{\overtwo}{\ensuremath{\overline{2}}}
\newcommand{\Thetastruct}[1]{\ensuremath{\Theta_{#1}}}

\newcommand{\Hstruct}[1]{\ensuremath{\mathcal{H}(#1)}}
\newcommand{\Zstruct}[1]{\ensuremath{Z({#1})}}
\newcommand{\dZstruct}[1]{\hat{Z}(#1)}
\newcommand{\Kstruct}[1]{\ensuremath{K({#1})}}
\newcommand{\Zln}{\Zstruct{\overln}}
\newcommand{\Zn}{\Zstruct{\overn}}
\newcommand{\Zl}{\Zstruct{\overl}}
\newcommand{\Zdelta}{\Zstruct{\delta}}
\newcommand{\dZln}{\dZstruct{\overln}}
\newcommand{\dZn}{\dZstruct{\overn}}
\newcommand{\dZl}{\dZstruct{\overl}}
\newcommand{\dZdelta}{\dZstruct{\delta}}
\newcommand{\Hln}{\Hstruct{\overln}}
\newcommand{\Hn}{\Hstruct{\overn}}

\newcommand{\Hdelta}{\Hstruct{\delta}}

\newcommand{\Kdelta}{\Kstruct{\delta}}
\newcommand{\Thetaln}{\Thetastruct{\overln}}
\newcommand{\Thetan}{\Thetastruct{\overn}}

\newcommand{\Thetadelta}{\Thetastruct{\delta}}
\newcommand{\Mstruct}[1]{\mathcal{M}_{#1}}
\newcommand{\Mln}{\Mstruct{\overln}}
\newcommand{\Mn}{\Mstruct{\overn}}

\newcommand{\Mdelta}{\Mstruct{\delta}}
\newcommand{\Mbarstruct}[1]{\overline{\mathcal{M}}_{#1}}
\newcommand{\Mbarln}{\Mbarstruct{\overln}}
\newcommand{\Mbarn}{\Mbarstruct{\overn}}

\newcommand{\Mbardelta}{\Mbarstruct{\delta}}
\newcommand{\Ztwo}{\Zstruct{\overtwo}}
\newcommand{\dZtwo}{\dZstruct{\overtwo}}
\newcommand{\Az}{A^0}
\newcommand{\pionetilde}{\tilde{\pi}_1}
\newcommand{\pitwotilde}{\tilde{\pi}_2}
\title{Computing modular correspondences for abelian varieties}
\author{Jean-Charles Faug\`ere\inst{1}, David Lubicz\inst{2,3}, Damien
Robert\inst{4}}

\institute{
INRIA, Centre Paris-Rocquencourt, SALSA Project\\
UPMC, Univ Paris 06, LIP6\\
CNRS, UMR 7606, LIP6\\
UFR Ing\'enierie 919,
LIP6 Passy Kennedy,
Boite courrier 169,\\
4, place Jussieu,
F-75252 Paris Cedex 05
\and
C\'ELAR,
BP 7419,
F-35174 Bruz
\and
IRMAR, Universt\'e de Rennes 1,
Campus de Beaulieu,
F-35042 Rennes
\and
LORIA, CACAO Project\\ 
Campus Scientifique\\
BP 239\\
54506 Vandoeuvre-l\`es-Nancy Cedex\\
}

\begin{document}

\maketitle

\bigskip
\begin{abstract}
\noindent
The aim of this paper is to give a higher dimensional equivalent of
the classical modular polynomials $\Phi_\ell(X,Y)$. If $j$ is the
$j$-invariant associated to an elliptic curve $E_k$ over a field $k$
then the roots of $\Phi_\ell(j,X)$ correspond to the $j$-invariants of
the curves which are $\ell$-isogeneous to $E_k$.  Denote by $X_0(N)$
the modular curve which parametrizes the set of elliptic curves
together with a $N$-torsion subgroup. It is possible to interpret
$\Phi_\ell(X,Y)$ as an equation cutting out the image of a certain
modular correspondence $X_0(\ell) \rightarrow X_0(1) \times X_0(1)$ in
the product $X_0(1) \times X_0(1)$.

Let $g$ be a positive integer and $\overn \in \N^g$.  We are
interested in the moduli space that we denote by
$\Mn$ of abelian varieties of dimension $g$
over a field $k$ together with an ample symmetric line bundle
$\pol$ and a symmetric theta structure of type $\overn$. If $\ell$ is a
prime and let $\overl=(\ell, \ldots , \ell)$, there exists a
modular correspondence $\Mln \rightarrow
\Mn \times \Mn$. We give
a system of algebraic equations defining the image of this modular
correspondence.

We describe an algorithm to solve this system of algebraic equations
which is much more efficient than a general purpose Gr\"obner basis
algorithm. As an application, we explain how this algorithm can be
used to speed up the initialisation phase of a point counting
algorithm.
\end{abstract}

\section{Introduction}
\label{s_intro}
The aim of this paper is to give a higher dimensional equivalent of
the classical modular polynomials $\Phi_\ell(X,Y)$. We recall that
$\Phi_\ell(X,Y)$ is a polynomial with integer coefficients and that if
$j$ is the $j$-invariant associated to an elliptic curve $E_k$ over a
field $k$ then the roots of $\Phi_\ell(j,X)$ correspond to the
$j$-invariants of elliptic curves which are $\ell$-isogeneous to
$E_k$. These modular polynomials have important algorithmic
applications. For instance, Atkin and Elkies (see \cite{MR1486831})
take advantage of the modular parametrisation of $\ell$-torsion
subgroups of an elliptic curve to improve the original point counting
algorithm of Schoof \cite{Schoof1}.

In \cite{MR2001j:11049}, Satoh has introduced an algorithm to count
the number of rational points of an elliptic curve $E_k$ defined over
a finite field $k$ of small characteristic $p$ which rely on the
computation of the canonical lift of the $j$-invariant of $E_k$. Here
again it is possible to improve the original lifting algorithm of
Satoh \cite{MR1895422,MR2293798} by solving over the $p$-adics an
equations given by the modular polynomial $\Phi_p(X,Y)$.

This last algorithm has been improved by Kohel in \cite{MR2093256}
using the notion of oriented modular correspondence. For $N\in \N^*$,
the modular curve $X_0(N)$ parametrizes the set of isomorphism classes
of elliptic curves together with a $N$-torsion subgroup. For instance,
the curve $X_0(1)$ is just the line of $j$-invariants. Let $p$ be
prime to $N$. A rational map of curves $X_0(pN) \rightarrow X_0(N)
\times X_0(N)$ is an oriented modular correspondence if the image of
each point represented by a pair $(E,G)$ where $G$ is a subgroup of
order $pN$ of $E$ is a couple $((E_1, G_1), (E_2, G_2))$ with $E_1=E$
and $G_1$ is the unique subgroup of index $p$ of $G$, and $E_2=E/H$
where $H$ is the unique subgroup of order $p$ of $G$. In the case that
the curve, $X_0(N)$ has genus zero, the correspondence can be
expressed as a binary equation $\Phi(X,Y)=0$ in $X_0(N) \times X_0(N)$
cutting out a curve isomorphic
to $X_0(pN)$ inside the product. For instance, if one consider the
oriented correspondence $X_0(\ell) \rightarrow X_0(1) \times X_0(1)$
for $\ell$ a prime number then the polynomial defining its image in
the product is the modular polynomial $\Phi_\ell(X,Y)$.

In this paper, we are interested in the computation of an analog of
oriented modular correspondences for higher dimensional abelian
varieties over a field $k$.  We use a moduli space which is different
from the one of \cite{MR2093256}. We fix an integer $g>0$ for the rest
of the paper. In the following if $n$ is an integer, $\overn$
denotes the element $(n, \ldots , n) \in \Z^g$.  We consider the set
of triples of the form $(A_k, \pol, \Thetan)$ where
$A_k$ is a $g$ dimensional abelian variety equipped with a symmetric
ample line bundle $\pol$ and a symmetric theta structure $\Thetan$
of type $\overn$.  Such a triple is called an abelian variety
with a $\overn$-marking. To a triple $(A_k, \pol,
\Thetan)$, one can associate following \cite{MR34:4269}
its theta null point. The locus of theta null points corresponding to
the set of abelian varieties with a $\overn$-marking is a
quasi-projective variety $\Mn$.  Moreover, it
is proved in \cite{MR36:2621} that if $8|n$ then
$\Mn$ is a classifying space for abelian
varieties with a $\overn$-marking.  We would like to compute
oriented modular correspondences in $\Mn$.

For this, let $(A_k, \pol, \Thetaln)$ be an abelian
variety with a $(\overln)$-marking. We suppose that $\ell$
and $n$ are relatively prime. From the theta structure
$\Thetaln$, we deduce a decomposition of the kernel
of the polarization $K(\pol)=K_1(\pol) \times K_2(\pol)$ into maximal
isotropic subspaces for the commutator pairing associated to $\pol$.
Let $K(\pol)[\ell]=K_1(\pol)[\ell] \times K_2(\pol)[\ell]$ be the
induced decomposition of the $\ell$-torsion part of $K(\pol)$. Let
$B_k$ be the quotient of $A_k$ by $K_2(\pol)[\ell]$ and $C_k$ be the
quotient of $A_k$ by $K_1(\pol)[\ell]$. In this paper, we show that
the theta structure of type $\overln$ of $A_k$ induces in a
natural manner theta structures of type $\overn$ on $B_k$ and
$C_k$. As a consequence, we obtain a modular correspondence,
$\Mln \rightarrow
\Mn \times \Mn$.  In the
projective coordinate system provided by theta constants, we give
a system of equations for the image of $\Mln$ in the
product $\Mn \times \Mn$
as well as an efficient algorithm to solve this system.

This paper is organized as follows. In Section~\ref{s_notations} we
recall some basic definitions and properties about algebraic theta
functions. In Section~\ref{s_isogenies}, we define formally the
modular correspondence, and then in Section~\ref{s_modcorr} we give
explicit equations for the computation of this correspondence. In
particular, we define a polynomial system (the
equations of the image of $\Mln$), which solutions give theta null
points of isogeneous varieties. In Section~\ref{s_solutions}, we
describe the geometry of these solutions. The last Section is devoted to
the description of a fast algorithm compute the solutions.

\section{Some notations and basic facts} \label{s_notations}
In this section, we fix some notations for the rest of the paper and
recall well known results on abelian varieties and theta structures.

Let $A_k$ be a $g$ dimensional abelian variety over a field $k$. Let
$\pol$ be a degree $d$ ample symmetric line bundle on $A_k$.  From
here, we suppose that $d$ is prime to the characteristic of $k$ or
that $A_k$ is ordinary.  Denote by $K(\pol)$ the kernel of the
polarization $\pol$ and by $G(\pol)$ the theta group (see
\cite{MR34:4269}) associated to $\pol$. The theta group $G(\pol)$ is
by definition the set of pairs $(x,\psi)$ where $x$ is a geometric point
of $K(\pol)$ and $\psi$ is an isomorphism of line bundles $\psi : \pol
\rightarrow \tau^*_x \pol$ together with the composition law $(x,
\psi)\circ (y, \phi)=(x+y, \tau_y^* \psi \circ \phi)$.  Let
$\delta=(d_1, \ldots , d_g)$ be a finite sequence of integers such
that $d_i | d_{i+1}$, we consider the finite group scheme $\Zdelta =
(\xz/d_1 \xz)_k \times_k \ldots \times_k (\xz/d_g \xz)_k$ with
elementary divisors given by $\delta$.  For a well chosen unique
$\delta$, the finite group scheme $\Kdelta = \Zdelta \times \dZdelta$
(where $\dZdelta$ is the Cartier dual of $\Zdelta$) is isomorphic to
$K(\pol)$ (see \cite{MR0282985}).  The Heisenberg group of type
$\delta$ is the scheme $\mathcal{H}(\delta) = \xg_{m,k} \times
Z(\delta) \times \hat{Z}(\delta)$ together with the group law defined
on geometric points by $(\alpha, x_1, x_2).(\beta, y_1, y_2) =
(\alpha.\beta.y_2(x_1), x_1 + y_1, x_2 + y_2)$.  We recall
\cite{MR34:4269} that a theta structure $\Thetadelta$ of type $\delta$
is an isomorphism of central extension from $\Hdelta$ to $G(\pol)$
fitting in the following diagram:
\begin{equation}\label{thetastructure}
\xymatrix{
0 \ar[r] & \xg_{m,k} \ar[r] \ar@{=}[d] & \Hdelta \ar[r] \ar[d]^{\Thetadelta}
& \Kdelta \ar[r] \ar[d]^{\overline{\Theta}_\delta} & 0 \\
0 \ar[r] & \xg_{m,k} \ar[r] & G(\pol) \ar[r]^\kappa
& K(\pol) \ar[r] & 0 
  }.
\end{equation}

We note that $\Thetadelta$ induces an isomorphism, denoted
$\overline{\Theta}_\delta$ in the preceding diagram, from $\Kdelta$
into $K(\pol)$ and as a consequence a decomposition $K(\pol)=K_1(\pol)
\times K_2(\pol)$ where $K_2(\pol)$ is the Cartier dual of
$K_1(\pol)$. The data of a triple $(A_k, \pol, \Thetadelta)$ defines
a basis of global sections of $\pol$ that we denote $(\vartheta_i)_{i
  \in \Zdelta}$ and as a consequence an morphism of $A_k$ into
$\proj_k^{d-1}$ where $d=\prod_{i=1}^g d_i$ is the degree of $\pol$.
We briefly recall the construction of this basis.  We recall \cite[pp.
291]{MR34:4269} that a level subgroup $\tilde{K}$ of $G(\pol)$ is a
subgroup such that $\tilde{K}$ is isomorphic to its image by $\kappa$
in $K(\pol)$ where $\kappa$ is defined in (\ref{thetastructure}).
We define the maximal level subgroups $\tilde{K}_1$ over $K_1(\pol)$
and $\tilde{K}_2$ over $K_2(\pol)$ as the image by $\Thetadelta$ of
the subgroups $(1,x,0)_{x \in \Zdelta}$ and $(1,0,y)_{y \in
  \dZdelta}$ of $\Hdelta$.  Let $\Az_k$ be the
quotient of $A_k$ by $K_2(\pol)$ and $\pi : A_k \rightarrow \Az_k$ be
the natural projection.  By the descent theory of Grothendieck, the
data of $\tilde{K}_2$ is equivalent to the data of a couple $(\pol_0,
\lambda)$ where $\pol_0$ is a degree one ample line bundle on $\Az_k$
and $\lambda$ is an isomorphism $\lambda : \pi^*(\pol_0) \rightarrow
\pol$. Let $s_0$ be the unique global section of $\pol_0$ up to a
constant factor and let $s = \lambda(\pi^*(s_0))$. We have the
following proposition (see \cite{MR34:4269})
\begin{proposition}\label{prop1}
  For all $i \in \Zdelta$, let $(x_i, \psi_i)=
  \Thetadelta((1,i,0))$.  We set $\vartheta^{\Thetadelta}_i =
  (\tau_{-x_i}^*\psi_i(s))$.  
  The elements
  $(\vartheta^{\Thetadelta}_i)_{i \in \Zdelta}$ form a basis of the
  global sections of $\pol$ which is uniquely determined up to a
  multiplication by a factor independent of $i$ by the data of
  $\Thetadelta$.
\end{proposition}
If no ambiguity is possible, we let
$\vartheta^{\Thetadelta}_i=\vartheta_i$ for $i\in \Zdelta$.

The image of the zero point $0$ of $A_k$ by the morphism provided
by $\Thetadelta$, which has homogeneous coordinates
$(\vartheta_i(0))_{i \in \Zdelta}$, is by definition the theta null
point associated to $(A_k, \pol, \Thetadelta)$.  If $\Thetadelta$ is
symmetric \cite[pp. 317]{MR34:4269}, we say that $(A_k, \pol,
\Thetadelta)$ is an abelian variety with a $\delta$-marking. The locus
of the theta null points associated to abelian varieties with a
$\delta$-marking is a quasi-projective variety denoted $\Mdelta$.

Let $(A_k, \pol, \Thetadelta)$ be an abelian variety with a
$\delta$-marking. We recall that the natural action of $G(\pol)$ on
the global sections of $\pol$ is given by $(x, \psi) . f =
\tau_{-x}^{*}\psi(f)$ for $f \in \Gamma(\pol)$ and $(x, \psi) \in
G(\pol)$. There is an action of $\Hdelta$ on $(\vartheta_i)_{i \in
  \Zdelta}$ given by:
\begin{equation}
  (\alpha,i,j). \vartheta_{k} =\alpha
e_\delta(k+i,-j) \vartheta_{k+i},
\label{eq_actiontheta}
\end{equation}
for $(\alpha,i,j) \in \Hdelta$ and $e_\delta$ the commutator pairing
on $\Kdelta$, which is compatible via $\Thetadelta$ with the natural
action of $G(\pol)$ on $(\vartheta_i)_{i \in \Zdelta}$. Using
(\ref{eq_actiontheta}), one can compute the coordinates in the
projective system given by the $(\theta_i)_{i \in \Zdelta}$ of any
point of $K(\pol)$ from the data of the theta null point associated to
$(A_k, \pol, \Thetadelta)$.

Let $\delta=(\delta_1, \ldots, \delta_g) \in \N^g$ and
$\delta'=(\delta'_1, \ldots , \delta'_g) \in \N^g$, $\delta|\delta'$ means that for $i=1,
\ldots , g$, $\delta_i | \delta'_i$. If $n \in \N$, $n | \delta$ means that
$(n,\ldots,n) \in \N^g | \delta$. If $\delta | \delta'$ we have the usual
embedding 
\begin{equation}
  i: \Zdelta \to \Zstruct{\delta'}, (x_i)_{i \in \{1, \ldots ,g\}} \mapsto
  (\delta'_i/\delta_i . x_i) 
  \label{eq_class_embedding}
\end{equation}

A basic ingredient of our algorithm is given by 
the Riemann relations which are algebraic relations satisfied by the
theta null values if $4 | \delta$.
\begin{theorem}\label{riemannquad}
  Denote by $\dZtwo$ the dual group of $\Ztwo$.  Let $(a_i)_{i \in
    \Zdelta}$ be the theta null points associated to an abelian
  variety with a $\delta$-marking $(A_k, \pol, \Thetadelta)$ where $2
  | \delta$.  For all $x,y,u,v \in \Zstruct{2 \delta}$ which are
  congruent modulo $\Ztwo$, and all $\chi \in \dZtwo$, we have
\begin{eqnarray*}
\big(\sum_{t \in \Ztwo} \chi(t) \vartheta_{x+y+t} \vartheta_{x-y+t}\big).\big(\sum_{t \in \Ztwo} \chi(t)
a_{u+v+t} a_{u-v+t}\big)= \\
=\big(\sum_{t \in \Ztwo} \chi(t) \vartheta_{x+u+t} \vartheta_{x-u+t}\big).\big(\sum_{t \in \Ztwo} \chi(t)
a_{y+v+t} a_{y-v+t}\big). 
\end{eqnarray*}

Here we embed $\Ztwo$ into $\Zdelta$ and $\Zdelta$ into $\Zstruct{2
\delta}$ using~\eqref{eq_class_embedding}.
\end{theorem}
It is moreover proved in \cite{MR34:4269} that if $4|\delta$ the image
of $A_k$ by the projective morphism defined by $\Thetadelta$ is the
closed subvariety of $\proj_k^{d-1}$ defined by the homogeneous ideal
generated by the relations of Theorem \ref{riemannquad}.

A consequence of Theorem~\ref{riemannquad} is the fact that if
$4|\delta$, from the knowledge of a valid theta null point $(a_i)_{i
  \in \Zdelta}$, one can recover a couple $(A_k, \pol)$ from which it
comes from. In fact, the abelian variety $A_k$ is defined by the
homogeneous equations of Theorem~\ref{riemannquad}. Moreover, from the
knowledge of the projective embedding of $A_k$, one recover immediately
$\pol$ by pulling back the sheaf $\mathcal{O}(1)$ of the projective
space.

An immediate consequence of the preceding theorem is the
\begin{theorem}\label{thetarel}
  Let $(a_i)_{i \in \Zdelta}$ be the theta null point associated to
  an abelian variety with a $\delta$-marking $(A_k, \pol,
  \Thetadelta)$ where $2 | \delta$.  For all $x,y,u,v \in Z(2 \delta)$
  which are congruent modulo $\Zdelta$, and all $\chi \in \dZtwo$, we
  have
\begin{eqnarray*}
\big(\sum_{t \in \Ztwo} \chi(t) a_{x+y+t} a_{x-y+t}\big).\big(\sum_{t \in \Ztwo} \chi(t)
a_{u+v+t} a_{u-v+t}\big)= \\
=\big(\sum_{t \in \Ztwo} \chi(t) a_{x+u+t} a_{x-u+t}\big).\big(\sum_{t \in \Ztwo} \chi(t)
a_{y+v+t} a_{y-v+t}\big). 
\end{eqnarray*}
As $\Thetadelta$ is symmetric, the theta constants also satisfy the
additional symmetry relations $a_i=a_{-i}$, $i \in \Zdelta$.
\end{theorem}

The Theorem \ref{thetarel} gives equations satisfied by the
theta null points of abelian varieties together with a
$\delta$-marking. Let $\Mbardelta$ be the projective variety
over $k$ defined by the symmetry relations together with the relations
from theorem~\ref{thetarel}. 
Mumford proved in \cite{MR36:2621} the following

\begin{theorem}
 Suppose that $8 | \delta$. Then
 \begin{enumerate}
   \item 
     $\Mdelta$ is a classifying space for abelian varieties with a
     $\delta$-marking: to a theta null point corresponds a unique
     triple $(A_k, \pol, \Thetadelta)$.
   \item 
     $\Mdelta$ is an open subset of $\Mbardelta$.
 \end{enumerate}
  \label{th_fond_mumford}
\end{theorem}

A geometric point $P$ of $\Mbardelta$ is called a theta constant. If a
theta constant $P$ is in $\Mdelta$ we say that $P$ is a valid theta
null point, otherwise we say that $P$ is a degenerate theta null
point.

\begin{remark}
  As the results of Section~\ref{s_solutions} show, $\Mbardelta$ may not be
  a projective closure of $\Mdelta$. Nonetheless, 
  every degenerate theta null point
  can be obtained from a valid theta null point by a ``degenerate'' group
  action (see the discussion after Proposition~\ref{lem_damien}), hence the notation.
\end{remark}

\section{Theta null points and isogenies} \label{s_isogenies}
In this section, we are interested in the following situation.  Let
$\ell$ and $n$ be relatively prime integers and suppose that $n$ is
divisible by $2$.  Let $(A_k, \pol, \Thetaln)$ be a $g$-dimensional
abelian
variety together with a $(\overln)$-marking. We recall that
the theta structure $\Thetaln$ induces a
decomposition of the kernel of the polarization
\begin{equation}\label{decomp}
K(\pol)=K_1(\pol) \times K_2(\pol)
\end{equation}
into maximal isotropic subgroups for the commutator pairing associated
to $\pol$. Let $K$ be a maximal isotropic $\ell$-torsion subgroup of
$K(\pol)$ compatible with the decomposition (\ref{decomp}). There are
two possible choices for $K$, one contained in $K_1(\pol)$, the other
one in $K_2(\pol)$. In the next paragraph, we explain that a choice of
$K$ determines a certain abelian variety together with a
$\overn$-marking. The main results of this Section are Corollary
\ref{coro1} and Proposition \ref{prop3} which explain how to compute
the theta null points associated to the abelian variety together with
a $\overn$-marking defined by a choice of $K$.

Let $X_k$ be the quotient of $A_k$ by $K$ and let $\pi : A_k
\rightarrow X_k$ be the natural projection.  Let $\kappa : G(\pol)
\rightarrow K(\pol)$ be the natural projection deduced from the
diagram (\ref{thetastructure}). As $K$ is a subgroup of $K(\pol)$, we
can consider the subgroup $G$ of $G(\pol)$ defined as
$G=\kappa^{-1}(K)$.  Let $\tilde{K}$ be the level subgroup of
$G(\pol)$ defined as the intersection of $G$ with the image of
$(1,x,y)_{(x,y) \in \Zln \times \dZln} \subset \Hln$ by
$\Thetaln$. By the descent theory of Grothendieck,
we know that the data of $\tilde{K}$ is equivalent to the data of a
line bundle $\xpol$ on $X_k$ and an isomorphism $\lambda :
\pi^*(\xpol) \rightarrow \pol$.

Now, we explain that the $(\overln)$-marking on $A_k$
induces a $\overn$-marking on $X_k$. Let $G^*(\pol)$ be the
centralizer of $\tilde{K}$ in $G(\pol)$.  Applying \cite[Proposition 2 pp. 291]{MR34:4269}
, we obtain an isomorphism
\begin{equation}\label{iso1}
G^*(\pol)/ \tilde{K} \simeq G(\xpol)
\end{equation}
and as a consequence a natural projection $q : G^*(\pol) \rightarrow
G(\xpol)$.

As $\Hn$ is generated by the subgroups
$1_{\xg_m} \times \Zn \times 0_{\dZn}$
and $1_{\xg_m} \times 0_{\Zn} \times
\dZn$, in order to define a theta structure
$\Thetan : \Hn \rightarrow
G(\xpol)$, it is enough to give morphisms $1_{\xg_m} \times
\Zn \times 0_{\dZn} \rightarrow G(\xpol)$
and $1_{\xg_m} \times 0_{\Zn} \times \dZn
\rightarrow G(\xpol)$. Let $Z^*(\overln)$ be such that
$1_{\xg_m} \times Z^*(\overln) \times 0_{\dZln}
=\Thetaln^{-1}(G^*(\pol)) \cap \Zln$
and let $\hat{Z}^*(\overln)$ be such that $1_{\xg_m} \times
0_{\Zln} \times
\hat{Z}^*(\overln)=\Thetaln^{-1}(G^*(\pol)) \cap
\dZln$.

As $\hat{Z}^*(\overln)$ is in the orthogonal of
$\overline{\Theta}^{-1}_{\overln}(K)$ for the commutator
pairing, we have $\hat{Z}^*(\overln)=\dZln$ or 
$\hat{Z}^*(\overln)=\dZn$
depending on the choice of $\tilde{K}$. In any case, there exists
a natural projection $p : \hat{Z}^*(\overln) \rightarrow
\dZn$.  In the same way, $Z^*(\overln)=\Zln$ or 
$Z^*(\overln)=\Zn$ and there is a natural injection $i :
\Zn \rightarrow Z^*(\overln)$.

We can define $\Thetan$ as the unique theta structure
for $\xpol$ such that the following diagrams are commutative
\begin{equation}\label{eq1}
\xymatrix{
(1,0,y)_{y \in \hat{Z}^*(\overln)} \ar[d]^{\tilde{p}}
\ar[r]^{\Thetaln} & G^*(\pol) \ar[d]^q \\ 
(1,0,y)_{y \in \hat{Z}(\overn)} \ar[r]^{\Thetan} & G(\xpol) 
},
\end{equation}

\begin{equation}\label{eq2}
  \xymatrix{
    (1,x,0)_{y \in Z^*(\overln)} \ar[r]^{\Thetaln} & G^*(\pol) \ar[d]^q \\
    (1,x,0)_{y \in \Zn}  \ar[u]^{\tilde{i}} \ar[r]^{\Thetan} & G(\xpol) 
  },
\end{equation}
where $\tilde{i}$ is deduced from $i$ and $\tilde{p}$ is deduced from
$p$. Using the fact that $\Thetaln$ is symmetric, it is easy to see
that $\Thetan$ is also symmetric.

We say that the theta structures $\Thetaln$ and
$\Thetan$ are $\pi$-compatible (or compatible) if the diagrams
(\ref{eq1}) and (\ref{eq2}) commute.

Let $K_1$ and $K_2$ be the maximal $\ell$-torsion subgroups of
respectively $K_1(\pol)$ and $K_2(\pol)$. By taking $K=K_2$ and
$K=K_1$ in the preceding construction, we obtain respectively $(B_k,
\ppol, \Thetan)$ and $(C_k, \pppol, \Theta'_{\overn})$ two abelian varieties
with a $\overn$-marking. As a consequence, we have a well defined
modular correspondence
\begin{equation}\label{modular}
\Phi_\ell: \Mln \rightarrow \Mn
\times \Mn.
\end{equation}

Let $\pi : A_k \rightarrow B_k$ and $\pi': A_k \rightarrow C_k$ be the
isogenies deduced from the construction.  Let $[\ell]$ be the isogeny
of multiplication by $\ell$ on $B_k$ and let $\hat{\pi} : B_k
\rightarrow A_k$ be the isogeny such that $[\ell]= \pi \circ
\hat{\pi}$. From the symmetry of $\pol$ we deduce that $\ppol$ is
symmetric and by applying the formula of \cite[pp. 289]{MR34:4269}, we
have $[\ell]^* \ppol= \ppol^{\ell^2}$.  The following diagram shows
that $C_k$ is obtain by quotienting $B_k$ by a maximal isotropic
subgroup of $(B_k, \ppol^{\ell^2})$ of order $\ell^{2g}$.

\begin{equation}
\xymatrix{
B_k \ar[dd]^{[\ell]} \ar[dr]^{\hat{\pi}} & & \\
& A_k \ar[dl]^{\pi} \ar[dr]^{\pi'} & \\
B_k & & C_k \\
}.
  \label{eq_diag1}
\end{equation}

The following two propositions explain the relation between the theta
null point of $(A_k, \pol, \Thetaln)$ and the theta
null points of $(B_k, \ppol, \Thetan)$ and
$(C_k, \pppol, \Theta'_{\overn})$.
Keeping the notations of the previous paragraph, we have
\begin{proposition}
  Let $(A_k, \pol, \Thetaln)$, $(B_k, \ppol,
  \Thetan)$ and $\pi:A_k \rightarrow B_k$ be defined as
  above. There exists a constant factor $\omega \in \overline{k}$ such
  that for all $i \in \Zn$, we have
  $\pi^{*}(\vartheta_i^{\Thetan}) = \omega
  \vartheta_i^{\Thetaln}$. In this last relation,
  $\Zn$ is identified as a subgroup of $\Zln$ 
  via the map $x \mapsto \ell x$. 
\end{proposition}

\begin{proof}
  This proposition is a particular case of the isogeny theorem 
\cite[Th. 4]{MR34:4269} but we give here a direct proof.

Let $\tilde{K}_A$ be the level subgroup of $G(\pol)$ defined by the
image of $(1,0,y)_{y \in \dZln}$ by $\Thetaln$ and let $K_A$ be the
subgroup of $A_k$ which is the image of $(0,y)_{y \in \dZln}$ by
$\overline{\Theta}_{\overln}$. Let $D_k$ be the quotient of $A_k$ by
$K_A$ and $\pi_A : A_k \rightarrow D_k$ the natural projection. The
data of $\tilde{K}_A$ gives a couple $(\pol_A, \lambda_A)$ where
$\pol_A$ is a degree one line bundle on $D_k$ and $\lambda_A$ is an
isomorphism $\lambda_A : \pi_A^*(\pol_A) \rightarrow \pol$. We recall
that $\tilde{K}$ be the level subgroup of $G(\pol)$ defined as the
intersection of $G=\kappa^{-1}(K)$ with the image of $(1,x,y)_{(x,y)
  \in \Zln \times \dZln} \subset \Hln$ by $\Thetaln$.

In the same manner, we can consider $\tilde{K}_B$ the level subgroup
of $G(\ppol)$ defined by the image of $(1, 0, y)_{y \in \dZn}$ by
$\Thetan$ and $K_B$ the subgroup of $B_k$ which is the image of
$(0,y)_{y \in \dZn}$ by $\overline{\Theta}_{\overn}$. By (\ref{eq1})
$K_B=\pi(K_A)$ and by definition of $\pi$ its kernel $K$ is contained
in $K_A$. We deduce that $D_k$ is the quotient of $B_k$ by $K_B$ and
$\pi_A = \pi_B \circ \pi$ where $\pi_B$ is the natural projection $B_k
\to D_k$. Because of (\ref{eq1}) and the fact that
$\hat{Z}^*(\overln)= \dZln$, we have an isomorphism $\tilde{K}_B\simeq
\tilde{K}_B / \tilde{K}$ and the data of $\tilde{K_B}$ gives a couple
$(\pol_B, \lambda_B)$ where $\lambda_B$ is an isomorphism $\lambda_B :
\pi_B^*(\pol_B) \rightarrow \ppol$ and we have $\pol_B=\pol_A$ and
$\lambda_A \circ \pi_A^* = \lambda \circ \pi^* \circ \lambda_B \circ
\pi_B^*$.

If $s_0$ is the unique global section of $\pol_A$ up to multiplication
by a constant factor, we have $\lambda_A(\pi_A^*(s_0))=
\lambda(\pi^*(\lambda_B(\pi_B^*(s_0))))$.  By definition,
$\vartheta_0^{\Thetan} = \lambda_B(\pi_B^*(s_0))$ and
$\vartheta_0^{\Thetaln} = \lambda_A(\pi_A^*(s_0))$. As a consequence,
there exists $\omega \in \overline{k}$ such that we have that
$\pi^*(\vartheta_0^{\Thetan}) = \omega \vartheta_0^{\Thetaln}$.

Let $s=\vartheta_0^{\Thetaln}$ and $s'=\vartheta_0^{\Thetan}$. We set
for all $i \in \Zln$, $(x_i, \psi_i)=\Thetaln((1,i,0))$ and for all $i
\in Z(n)$, $(x'_i, \psi'_i)=\Thetan((1,i,0))$. Then
$\pi^*(\vartheta_i^{\Thetan})=\pi^*(\psi'_i \tau^*_{-x'_i}
(s'))=\psi_i \tau^*_{-x_i} \pi^* (s')$ by the commutativity of
(\ref{eq2}). But we already know that $\pi^*(s')=\omega s$ and $\psi_i
\tau^*_{-x_i}(\omega s)=\omega \vartheta_i^{\Thetaln}$. This concludes
the proof.
\end{proof}

As an immediate consequence of the preceding proposition, we have
\begin{corollary}\label{coro1}
  Let $(A_k, \pol, \Thetaln)$ and $(B_k, \ppol, \Thetan)$ be defined
  as above. Let $(a_u)_{ u \in \Zln}$ and $(b_u)_{ u \in \Zn}$ be
  theta null points respectively associated to $(A_k, \pol, \Thetaln)$
  and $(B_k, \ppol, \Thetan)$. Considering $\Zn$ as a subgroup of
  $\Zln$ via the map $x \mapsto \ell x$, there exists a constant
  factor $\omega \in \overline{k}$ such that for all $u \in \Zn$, $b_u
  = \omega a_u$.
\end{corollary}

\begin{proposition}\label{prop3}
  Let $(A_k, \pol, \Thetaln)$ and $(C_k, \ppol, \Thetan)$ be defined
  as above. Let $(a_u)_{ u \in \Zln}$ and $(c_u)_{ u \in \Zn}$ be the
  theta null points respectively associated to $(A_k, \pol, \Thetaln
  )$ and $(C_k, \ppol, \Thetan)$. We have for all $u \in \Zn$,
\begin{equation}\label{eqiso}
c_u= \sum_{t \in \Zl} a_{u+t},
\end{equation}
where $\Zn$ and $\Zl$ are considered as subgroups of $\Zln$ via the
maps $j \mapsto \ell j$ and $j \mapsto n j$.
\end{proposition}
\begin{proof}
  The theta structure $\Thetaln$ (resp.  $\Theta'_{\overn}$) induces a
  decomposition of the kernel of the polarization $K(\pol)=K_1(\pol)
  \times K_2(\pol)$ (resp.  $K(\pppol)=K_1(\pppol) \times
  K_2(\pppol)$).  Denote by $K'$ the kernel of $\pi'$. We have that
  $K'$ is a subvariety of $K_1(\pol)$ and we have an isomorphism:
  $$\sigma : K_1(\pol)/ K' \rightarrow K_1(\pppol).$$

  The hypothesis of \cite[Th. 4]{MR34:4269} are then verified and
  Equation~\eqref{eqiso} is an immediate application of this theorem.
\end{proof}

\section{The image of the modular correspondence} \label{s_modcorr}
In this section, we use the results of the previous section in order
to give equations for the image of the modular correspondence
$\Phi_\ell$.

We let $(B_k, \ppol, \Thetan)$ be an abelian variety together with a
$\overn$-marking and denote by $(b_u)_{u \in \Zn}$ its associated
theta null point. Let $\nu$ be the $2$-adic valuation of $n$. Unless
specified, we shall assume that $\nu \geq 3$.  Let $\mathscr{C}$ be
the reduced subvariety of $\Mn \times \Mn$ which is the image of
$\Phi_\ell(\Mln)$ in $\Mn \times \Mn$ given on geometric points by
$\pi: (a_u)_{u \in \Zln} \mapsto ((a_u)_{u \in \Zn}, (\sum_{ t \in
  \Zl} a_{u+t})_{u \in \Zn})$.

Denote by $p_1$ (resp. $p_2$) the restriction to $\mathscr{C}$ of the
first (resp.  second) projection from $\Mn \times \Mn$ into $\Mn$, and
let $\pi_1=p_1 \circ \pi$, $\pi_2=p_2 \circ \pi$.  We would like to
compute the algebraic set $\pi_2(\pi_1^{-1}((b_u)_{u \in \Zn}))$.  We
remark that this question is the analog in our situation to the
computation of the solutions of the equation $\Phi_\ell(j, X)$ defined
from the modular polynomial and $j \in \overline{k}$ a certain
$j$-invariant.

Let $\proj_k^{\Zln}=\SProj(k[x_u|u \in \Zln])$ be the ambient projective
space of $\Mbarln$, and
let $I$ be the homogeneous ideal defining $\Mbarln$
which is spanned by the relations of Theorem~\ref{thetarel}, together with
the symmetry relations.
Let $J$ be the image of $I$ under the specialization map
\begin{eqnarray*}
k[x_u|u \in
\Zln] \rightarrow k[x_u|u \in
\Zln,n u \not=0], \quad x_u \mapsto \left\{
\begin{array}{l@{,\hsp}l}
{b_u} & \mathrm{if} \quad u \in \Zn \\
{x_u} & \mathrm{else}
\end{array} \right. .
\end{eqnarray*}
and let $V_J$ be the affine variety defined by $J$.

Let $\tilde{\pi}^0_1 : \proj_k^{\Zln} \rightarrow \proj_k^{\Zn}$ and
$\tilde{\pi}^0_2 : \proj_k^{\Zln} \rightarrow \proj_k^{\Zn}$ be the
morphisms of the ambient projective spaces respectively defined on
geometric points by $(a_u)_{u \in \Zln} \mapsto (a_u)_{u \in \Zn}$ and
$(a_u)_{u \in \Zln} \mapsto (\sum_{t \in \Zl} a_{u+t})_{u \in \Zn}$.
Clearly, $\pi_1$ and $\pi_2$ are the restrictions of $\tilde{\pi}^0_1$
and $\tilde{\pi}^0_2$ to $\Mln$.  The morphism $\tilde{\pi}^0_1$ (resp
$\tilde{\pi}^0_2$) restricts to a morphism $\pionetilde: \Mbarln \to
\Mbarn$ (resp $\pitwotilde: \Mbarln \to \Mbarn$). By definition of
$J$, we have $V_J=\tilde{\pi}_1^{-1}(b_{u})_{u \in \Zn}$.

Let $S=k[y_u, x_v | u \in \Zn, v \in \Zln]$, we can consider $J$ as a
subset of $S$ via the natural inclusion of $k[x_u|u \in \Zln]$ into
$S$.  Let $\mathcal{L}'$ be the ideal of $S$ generated by $J$ together
with the elements $y_u - \sum_{t \in \Zl} x_{u+t}$ and
$\mathcal{L}=\mathcal{L}' \cap k[y_u| u \in \Zn]$. Let
$V_{\mathcal{L}}$ be the subvariety of $\mathbb{A}^{\Zn}$ defined by
the ideal $\mathcal{L}$.  By the definition of $\mathcal{L}$,
$V_{\mathcal{L}}$ is the image by $\tilde{\pi}_2$ of the fiber $V_J$,
so that $V_{\mathcal{L}}= \tilde{\pi}_2(\tilde{\pi}_1^{-1}(b_{u})_{u
  \in \Zn})$.

\begin{proposition}
  Keeping the notations from above, let $(b_u)_{u \in \Zn}$ be the
  geometric point of $\Mn$ corresponding to $(B_k, \ppol, \Thetan)$ .
  The algebraic variety $V^0_{\mathcal{L}}=\pi_2(\pi_1^{-1}(b_u)_{u
    \in \Zn})$ has dimension $0$ and is isomorphic to a subvariety of
  $V_{\mathcal{L}}$.
\end{proposition}
\begin{proof}
  From the preceding discussion the only thing left to prove is that
  $V^0_{\mathcal{L}}$ has dimension $0$.  But this follows from the
  fact that the algebraic variety $V_J$ has dimension $0$ \cite{pc3}.
\end{proof}

From an algorithmic point of view, the hard part of this modular
correspondence is the computation of $V^0_J=\pi_1^{-1}((b_u)_{u \in
  \Zn})$, the set of points in $V_J$ that are valid theta null points.
We proceed in two steps. First we compute the solutions in $V_J$ using
a specialized Gr\"obner basis algorithm (Section~\ref{subsec:general}) and
then we detect the valid theta null points using the results of the
next section (see Theorem~\ref{main}). But at first we recall the
geometric nature of $V^0_J$ given by Section~\ref{s_isogenies}:

\begin{proposition}
  $V^0_J$ is the locus of theta null points $(a_u)_{u \in \Zln}$ in
  $\Mln$ such that if $(A_k,\pol,\Thetastruct{\overln})$ is the
  corresponding variety with a $(\overln)$-marking then $\Thetaln$ is
  compatible with the theta structure $\Thetan$ of $B_k$.
  \label{prop_geometrie} 
\end{proposition}

\begin{proof}
  Let $(a_u)_{u\in \Zln}$ be a geometric point of $V^0_J$.  Let
  $(A_k,\pol, \Thetastruct{\overln})$ be a corresponding variety with
  $(\overln)$-marking. If we apply the construction of
  Section~\ref{s_isogenies}, we get an abelian variety $(B'_k,\pol_0',
  \Thetastruct{\overn}')$ with a $\overn$-marking and an isogeny $\pi:
  A_k \to B'_k$ such that $\Thetaln$ is compatible with $\Thetan'$.
  By definition of $J$, Corollary~\ref{coro1} shows that the theta
  null point of $B'$ is $(b_u)_{u \in \Zn}$. As $\nu \geq 2$, by
  Proposition \ref{riemannquad} this means that $B' \simeq B$. Since
  $\nu \geq 3$, we even know by Theorem~\ref{th_fond_mumford} that the
  triples $(B'_k,\pol_0', \Thetastruct{\overn}')$ and $(B_k, \pol_0,
  \Thetastruct{\overn})$ are isomorphic, so that $\Thetaln$ is
  compatible with $\Thetan$.
\end{proof}

We say that the isogeny from Section~\ref{s_isogenies} $A_k \to B_k$
is the $\overl$-isogeny associated to the $(\overln)$-marking of
$A_k$.


\section{The solutions of the system}\label{s_solutions}
This section is devoted to the study of the geometric points of $V_J$.
Our aim is twofolds. First we need a way to identify degenerate theta
null points in $V_J$, and then we would like to know when two geometric
points in $V_J$ correspond to isomorphic varieties.

If $(a_u)_{u \in \Zln}$ is a valid theta null point, let $(A_k,\pol,
\Thetastruct{\overln})$ be the corresponding abelian variety with a
$(\overln)$-marking and denote by $\pi: A_k \to B_k$ the isogeny
defined in Section~\ref{s_isogenies}. From the knowledge of $(a_u)_{u
  \in \Zln}$, one can recover the coordinates the points of a maximal
$\ell$-torsion subgroup of $B_k$ of rank $g$.  Actually, even if
$(a_u)_{u \in \Zln}$ is not a valid theta null point it is possible to
associate to $(a_u)_{u \in \Zln}$ a set of $\ell$-torsion points $P_i
\in B_k$.  The main result of this section is Theorem~\ref{main} which
states that a geometric point of $V_J$ is non degenerate if and only if
the corresponding $P_i$ form a maximal subgroup of rank $g$ of the
$\ell$-torsion points of $B_k$. To prove this Theorem, we introduce an
action from the automorphisms of the theta group to the modular space
$\Mln$. Using Theorem~\ref{main} and this action, we make explicit the
structure of $V_J$: we explain when two valid points give isomorphic
varieties in Proposition~\ref{prop_structure}, and how to obtain every
degenerate points in the discussion following
Proposition~\ref{lem_damien}.

We start by making explicit the structure of the solutions of the
algebraic system defined by $J$.  For this let $\rho : \Zn \times \Zl
\rightarrow \Zln$ be the group isomorphism given by $(x,y) \mapsto
\ell x+ n y$. Denote by $I_{\Thetan}$ the ideal of $k[y_u | u \in
\Zn]$ for the theta structure $\Thetan$ generated by the equations of
Theorem \ref{riemannquad}. The homogeneous ideal $I_{\Thetan}$ defines
a projective variety $V_{I_{\Thetan}}$, isomorphic to $B_k$.

We have the following proposition \cite{pc3}:
\begin{proposition}
  \label{prop_pi}
  Let $(a_v)_{v \in \Zln}$ be a geometric point of $V_J$.  For any $i \in
  \Zl$ such that $(a_{\rho(j,i)})_{j \in \Zn} \ne (0, \ldots, 0)$, let
  $P_i$ be the geometric point, of $\proj_k^{\Zn}$ with homogeneous
  coordinates $(a_{\rho(j,i)})_{j \in \Zn}$.  Then for all $i \in \Zl$
  such that $P_i$ is well defined, $P_i$ is a $\ell$-torsion point of
  $V_{I_{\Thetan}}$.
\end{proposition}

The proof of the preceding proposition in \cite{pc3} proves moreover
that if we denote by $S$ the subset of $\Zl$ such that
$P_i$ is well defined for all $i \in S$, then $S$ is a subgroup of
$\Zl$, the set $\{ P_i , i \in S \}$ is a subgroup of the group of
$\ell$-torsion points of $V_{I_{\Thetan}}$ and the application $ i \in
S \to P_i \in B[\ell]$ is a group morphism. 

Suppose that $(a_v)_{v \in \Zln}$ is a valid theta null point. Let
$(A_k,\pol, \Thetastruct{\overln})$ be the corresponding abelian
variety with a $(\overln)$-marking and denote by $\pi: A_k \to B_k$
the isogeny defined in Section~\ref{s_isogenies}. 
We can consider $A_k$ as a closed subvariety of
$\proj_k^{Z(\overln)}$ via the morphism provided by
$\Thetastruct{\overln}$. Using the action~\eqref{eq_actiontheta} of
the theta group on $(a_v)_{v \in \Zln}$, one sees that for $i \in
\Zl$, the points with homogeneous coordinates $(a_{v+n i})_{v \in
  \Zln}$ form the isotropic (for the commutator pairing) $\ell$-torsion
  subgroup $K_1$ of $A_k$ (with the notations of Section~\ref{s_isogenies}).
  By definition of the isogeny $\pi$, we have
$\pi((a_{v+n i})_{v \in \Zln})=P_i$
and as a consequence $\pi(K_1)= \left\{ P_i, i \in \Zl \right\}$. We see that if
$(a_v)_{v \in \Zln}$ is a valid theta null point then the $(P_i)_{i
  \in \Zl}$ are well defined projective points which form a
maximal subgroup of rank~$g$ of $B_k[\ell]$.
Moreover, since the kernel of $\pi$ is $K_2$, 
$\pi(K_1)= \left\{P_i, i \in \Zl\right\}$ is the
kernel of the dual isogeny $\hat{\pi}: B_k \to A_k$. 

If $(a_v)_{v \in \Zln}$ is a general solution, it can happen that
certain of the $P_i$ are not well defined and as a consequence
$(a_v)_{v \in \Zln}$ is not a valid theta null point. But even if
every $P_i$ are well defined, $(a_v)_{v \in \Zln}$ need not be a valid
theta null point.  We need a criterion to identify the solutions of
$J$ which correspond to valid theta null points. From the discussion
of the preceding paragraph, we know that a necessary condition for a
solution $(a_v)_{v \in \Zln}$ of $J$ to be a valid theta null point is
that $(P_i)_{i \in \Zl}$ are all valid projective points which form a
subgroup of rank~$g$ if $B_k[\ell]$.  The Theorem \ref{main} asserts
that this necessary condition is indeed sufficient. In order to prove
this theorem, we have to study how a theta null point vary together
with a change of the theta structure.

We denote by $\Aut \Hdelta$ the group of automorphisms $\psi$ of
$\Hdelta$ inducing the identity on $\xg_{m,k}$: \[ \xymatrix{
0 \ar[r] & \xg_{m,k} \ar[r] \ar@{=}[d] & \Hdelta \ar[r]
\ar[d]^{\psi} & \Kdelta \ar[r] \ar[d]^{\overline{\psi}} & 0 \\
0 \ar[r] & \xg_{m,k} \ar[r] & \Hdelta \ar[r] & \Kdelta \ar[r] & 0 }. \]
Obviously, the set of all theta structures for $\pol$ is a principal
homogeneous space for the group $\Aut \Hdelta$ via the right action
$\Thetadelta.\psi = \Thetadelta \circ \psi$ for $\psi \in \Aut
\Hdelta$ and $\Thetadelta$ a theta structure. So we can identify $\Aut
\Hdelta$ with the group of automorphisms of theta structures. If
$\psi$ is such an automorphism, it induces an automorphism
$\overline{\psi}$ of $\Kdelta$. Denote by $Sp(\Kdelta)$ the group of
symplectic automorphisms of $\Kdelta$. The preceding diagram shows
that $\overline{\psi}$ is symplectic with respect to the commutator
pairing. Conversely, if $\overline{\psi} \in Sp(\Kdelta)$, we get an
element of $\Aut \Hdelta$ defined by $\psi: (\alpha,x,y) \mapsto
(\alpha, \psi(x), \psi(y))$. So the morphism $\abb{\Psi}{\Aut
  \Hdelta}{Sp(\Kdelta)}{\psi}{\overline{\psi}}$ has a section that we
denote by $\sigma$.  The kernel of $\Psi$ is given by automorphisms
preserving a symplectic basis and are determined by a choice of
level subgroups $\tilde{K}_1$ and $\tilde{K}_2$ over maximal isotropic
subspaces $K_1$ and $K_2$. It is well known \cite[pp. 162]{MR2062673}
that such choices are in bijection with elements $c \in K(\delta)$: we
map $c \in K(\delta)$ to the automorphism of $\Hdelta$ given by
\begin{equation}\label{sympl}
(\alpha,x,y) \mapsto (\alpha e_{\delta}(c,x+y),x,y).
\end{equation}

As a consequence, we get a split exact sequence
\begin{equation} \label{eq_actionseqeq}
\xymatrix{
  0 \ar[r]  & K(\delta) \ar[r]^\nu & \Aut \mathcal{H}(\delta) \ar[r] & Sp(K(\delta))
  \ar[r] \ar@/_1pc/[l]_{\sigma} &  0 \\
}.
\end{equation}
Suppose that $\Thetadelta$ is symmetric, an automorphism $\psi \in
\Aut \Hdelta$ is said to be symmetric if it commutes with the
symmetric action $(\alpha,x,y) \mapsto (\alpha,-x,-y)$ on $\Hdelta$.
We denote by $\Auts \Hdelta$ the group of symmetric automorphisms of
$\Hdelta$.  Obviously, an automorphism $\psi \in \Aut \Hdelta$ coming
from $c \in K(\delta)$ is symmetric if and only if $c \in
K(\delta)[2]$ the subgroup of $2$-torsion of~$K(\delta)$.

Now consider $(A_k, \pol, \Thetadelta)$ an abelian variety with a
$\delta$-marking and let $(\vartheta_i)_{i \in \Zdelta}$ be the
associated basis of global sections of $\pol$.  Note that if
$\overline{\psi}$ is a symplectic isomorphism of $K(\delta)$ then
$\psi=\sigma(\overline{\psi})$ is symmetric. 
We suppose that $\overline{\psi}(\dZdelta))=Z^{\psi} \times
\hat{Z}^{\psi}$, where $Z^{\psi} \subset \Zdelta$ and $\hat{Z}^{\psi} \subset
\dZdelta$.
Denote by
$(\tilde{\vartheta}_i)_{i \in \Zdelta}$ the basis of global sections of
$\pol$ associated to $(A_k, \pol, \Thetadelta.\psi)$.  In the
following, we give an explicit formula to obtain $(\tilde{\vartheta}_i)_{i
  \in \Zdelta}$ from the knowledge of $(\vartheta_i)_{i \in \Zdelta}$.

Let $A^0_k \simeq A_k/
\overline{\Theta_\delta}(\overline{\psi}(\hat{Z}(\delta)))$ and $\pi :
A_k \rightarrow A^0_k$ be the canonical map. The data of the maximal
level subgroup $\Theta_\delta(\psi((1,0,y)_{y \in \hat{Z}(\delta)}))$
is equivalent to the data of a line bundle $\pol_0$ on $A^0_k$ and an
isomorphism $\pi^*(\pol_0) \rightarrow \pol$.  Let $\tilde{s}_0$ be
the unique global section of $\pol_0$, we can apply the isogeny
theorem \cite[Th.  4]{MR34:4269} to obtain
\begin{gather}\label{action1}
 \tilde{\vartheta_0} = \lambda \pi^{*}(\tilde{s_0}) = \sum_{ i \in
 Z^{\psi} } \vartheta_i,
\end{gather}
for $\lambda \in k^*$.

Now by definition we have
\begin{equation}\label{action}
\tilde{\vartheta_i} = \psi((1,i,0)).\tilde{\vartheta_0}.
\end{equation}
where the dot product is the action (\ref{eq_actiontheta}).

By evaluating at $0$ the basis of global sections of $\pol$ in
(\ref{action}), we get an explicit description of the action of
$Sp(\Kdelta)$ on the geometric points of
$\mathcal{M}_{\delta}$. Actually the obtained formulas give a valid
action of $Sp(\Kdelta)$ on the geometric points of $\Mbardelta$.

Now, let $(a_u)_{u \in \Zln}$ be a geometric point of $V^0_{J}$. As
$\Auts \Hln$ acts on $\Mln$, we are interested in the subgroup
$\mathfrak{H}$ of $\Auts \Hln$ that leaves $(a_u)_{u \in \Zln}$ in
$V^0_J$.

\begin{lemma}
  Let $\psi \in \Aut \Hln$. We say that $\psi$ is compatible with $\Hn$ if
  it commutes with the morphisms $\tilde{p}$ and $\tilde{i}$
  from~\eqref{eq1} and~\eqref{eq2}. Then $\mathfrak{H}$
  is the subgroup of compatible symmetric automorphisms of $\Hln$. 
  In particular it does not depend on $(a_u)_{u \in \Zln}$ so it is also the
  subgroup of $\Auts \Hln$ that leaves $V^0_J$ invariant.
  \label{lem_automorphisme_compat}
\end{lemma}

\begin{proof}
  Let $(A,\pol,\Thetaln)$ be a triple corresponding to the theta
  null point $(a_u)_{u \in \Zln}$.  Let $\psi \in \Auts \Hln$, and
  $(a'_u)_{u \in \Zln} = \psi .(a_u)_{u \in \Zln}$.
  The Proposition~\ref{prop_geometrie} shows that $(a'_u)_{u \in \Zln}$ is
  in $V^0_J$ if and only if the associated theta structure $\Thetaln .
  \psi$ is compatible with the theta structure $\Thetan$ of $B$. But
  this means exactly that $\psi$ is compatible with $\Hn$.
\end{proof}

We can describe more precisely the action of $\mathfrak{H}$:
\begin{proposition}
  \label{lem_damien} 
  The action of
  $\mathfrak{H}$ on $V^0_J$ is generated by the
  actions given by
\begin{equation}
 (a_u)_{u \in \Zln} \mapsto (a_{{\psi_2}(u)})_{u \in \Zln},
  \label{eq_action1}
\end{equation}
where ${\psi_2}$ is an automorphism of $Z(\overln)$ fixing $Z(\overn)$ and
\begin{equation}
  (a_u)_{u \in \Zln} \mapsto \big(e_{\overln}({\psi_1}(u),u). a_u\big)_{u \in \Zln},
  \label{eq_action2}
\end{equation}
where $\psi_1$ is a ``symmetric'' morphism $\Zln \to \dZl \subset \dZln$
and $e_{\overln}$ is the commutator pairing on $\Hln$.
\end{proposition}
\begin{proof}
  Let $\psi \in \mathfrak{H}$. Since the exact sequence from
  equation~\eqref{eq_actionseqeq} splits, we only have to study the
  case where $\psi$ comes from a change of maximal level structure and
  the case where $\psi$ comes from a symplectic base change of
  $K(\delta)$. In the former case, let $c \in K(\delta)$ defining the
  symplectic base change by (\ref{sympl}).  Then $c \in K(\delta)[2]$
  since $\psi$ is symmetric and from the
  compatibility conditions $c \in \overline{\psi}(\hat{Z}(\overl))$.
  As $\ell$ is odd, we have $c=0$.

  In the latter case, $\overline{\psi}$ can be represented in a basis
  $(v_\kappa, \hat{v}_\kappa)_{\kappa \in \{1, \ldots , g \}}$ of
  $\Zln \times \dZln$ by a matrix $M[A,B,C,D]=\squaremat{A}{B}{C}{D}
  \in SP_{\delta}^{2g}(\Z)$.  Since $$K=\overline{\Theta}_{\overln}
  (\overline{\psi}(\dZl))  \subset \overline{\Theta}_{\overln}
  (\dZln),$$ we have $B=0$. So $D={}^{t}A^{-1}$ and we see that the action of
  $\mathfrak{H}$ is generated by the matrices
   \begin{enumerate}
   \item $M[A,B,C,D]$ such that $C$=0. Then $A$ is an automorphism and the
     compatibility condition implies that it must fix $\Zn$. Using
     (\ref{action1}) and (\ref{action}) this yields the
     action~\eqref{eq_action1}.
   \item $M[A,B,C,D]$ such that $A=\Id$. Then ${}^tC=C$. For $x \in \Zln$, we can write
     $\overline{\psi}((x,0))=(x,\psi_1(x))$. By looking at the
  conditions~\eqref{eq1} and~\eqref{eq2} we see that
\begin{equation}\label{eq_lem_damien}
 \overline{\psi}((x,y))-(x,y) \in \overline{\psi}(\dZl) \subset \dZl,
\end{equation}
for all $(x,y) \in Z^{*}(\overln) \times \hat{Z}^{*}(\overln)$.  Using
(\ref{eq_lem_damien}), we deduce that $\psi_1(x)$ is in $\dZl$. In this case,
we obtain the action~\eqref{eq_action2} following (\ref{action1}) and
(\ref{action}).
   \end{enumerate}
This completes the proof of the proposition.
\end{proof}

\begin{remark}
  \label{rem_sousgroupe}
  The action~\eqref{eq_action1} gives an automorphism of the $(P_i)_{i
    \in \Zl}$ while the action~\eqref{eq_action2} leaves the $(P_i)_{i
    \in \Zl}$ invariant. In fact by taking a basis of $\Zln$, we see
  that if $\zeta$ is a $(\ell n)^{th}$-root of unity, the
  actions~\eqref{eq_action2} are generated by
  \[ a_{(n_1,n_2,\ldots,n_g)} \mapsto \zeta^{\sum_{i,j \in [1,g]} a_{i,j} n_i n_j}
a_{(n_1,\ldots, n_g)}  \] 
where $(a_{i,j})_{i ,j \in [1, g]}$ is a symmetric matrix and 
$a_{i,j} \in \Z/n\Z \subset \Z/\ell n\Z$ (via $x
\mapsto \ell x$) for $i,j \in [1,g]$. 
So each coefficient of one $P_i$ is multiplied by the
same $\ell^{th}$-root of unity.
\end{remark}

From the preceding remark, we see that $\mathfrak{H}$ leaves $V_J$
invariant.  Now, let $\psi_2$ be a morphism of $\Zln$ fixing $\Zn$. Here
we do not require $\psi_2$ to be an isomorphism. We let $\psi_2$ act on
$V_J$ by
\[ 
 (a_u)_{u \in \Zln} \mapsto (a_{\psi_2(u)})_{u \in \Zln}
\]
Since $\psi_2$ fixes $\Zn \subset \Zln$, it fixes the $2$-torsion
points in $\Zln$, and it is easy to see that $(a_{\psi_2(u)})_{u \in
  \Zln}$ satisfies the equations of Theorem~\ref{thetarel} and the
symmetry relations. As a consequence, the point $(a_{\psi_2(u)})_{u
  \in \Zln}$ is in $\Mbarln$. Moreover, as $\psi_2$ fixes $\Zn$,
$(a_{\psi_2(u)})_{u \in \Zln}$ is a point in $V_J$, so we have a well
defined action extending that of the form (\ref{eq_action1}).

By acting on $V_J$ with a morphism of $\Zln$ fixing $\Zn$ which is not
an isomorphism, we obtain a point of $V_J$ which is degenerate: it is
a theta null point such that the associated points $P_i$ from
Proposition~\ref{prop_pi} are well defined but not distinct projective
points (so they do not form a rank $g$ $\ell$-torsion subgroup
of $B_k$).

There is another way to obtain degenerate theta null points in $V_J$.
Take any geometric point $(a_u)_{u \in \Zln} \in V_J$, and a subgroup $S$
of $\Zl$ (in particular $S$ is not empty). We define a new point
$(a'_u)_{u \in \Zln}$ where
\[ a'_{\rho(j,i)} = \begin{cases}
  a_{\rho(j,i)} & \text{if $i \in S$}, \\
  0             & \text{otherwise}.
\end{cases} \] Since $\ell$ is odd, it is easily seen that $(a'_u)_{u
  \in \Zln}$ is in general a degenerate point in $V_J$: the $P_i$ from
Proposition~\ref{prop_pi} are not defined when $i \not\in S$.

Now, we explain that combining the two methods described above, we
obtain all the degenerate theta null points of $V_J$. For this, let
$(a'_u)_{u \in \Zln}$ be a degenerate point of $V_J$. Let $S \subset
\Zl$ be the subgroup where the points of $\ell$-torsion $P'_i$, $i\in
S$ of Proposition~\ref{prop_pi} are well defined. The points $P'_i$
form a subgroup $S'$ of the $\ell$-torsion points of $B_k$, and $f:
S \to S', i \mapsto P'_i$ is a group morphism (which may not be an
isomorphism, since as $(a'_u)_{u \in \Zln}$ is degenerate the $P'_i$ are
not necessarily distinct). Now, we embed $S'$ into a maximal subgroup $T$ of
rank $g$ of $B_k[\ell]$, and extend $f$ to a morphism $\tilde{f}: \Zl \to T$
(for instance if $i \not\in \Zl$ then send $i$ to the neutral point
$P'_0$). We take an isomorphism $h$ between $\Zl$ and $T$. 
Theorem~\ref{main} that we prove later on shows that
there exists a geometric point $(a_u)_{u \in \Zln} \in V_J^0$ such that
the corresponding group morphism $i \in \Zl \mapsto P_i$ is $h$.
Now
take $\psi_2$ to be the morphism of $\Zln=\Zl \times \Zn$ which is the
identity on $\Zn$ and $h^{-1}\tilde{f}$ on $\Zl$.  Consider the point $(a_{\psi_2(u)})_{u
  \in \Zln}$ 
with the coefficients $\rho(j,i), i \not\in S$ taken to
be $0$. Then it has exactly the same defined points $P'_i$ as
$(a'_u)_{u \in \Zln}$. The next lemma shows that it is the same point
as $(a'_u)_{u \in \Zln}$ up to an action of the
form~\eqref{eq_action2}.

We remark that the degenerate points in $V_J$ are exactly the points where
the action of $\mathfrak{H}$ is not free: if $(a_u)_{u \in \Zln}$ is a
degenerate point such that the corresponding $P_i$ are not all well defined,
then there is an action of the form~\eqref{eq_action2} giving the same point. 
If the $P_i$ are well defined but do not form a maximal subgroup, then this
time there
is an action of the form~\eqref{eq_action1} giving the same point.



By Remark~\ref{rem_sousgroupe} we know that if $(a_u)_{u \in \Zln}$ is a
theta null point giving the associated group $\{ P_i, i \in \Zl \}$, then the
points $\psi.(a_u)_{u \in \Zln}$ where $\psi \in \mathfrak{H}$ give the same
associated group. In fact the converse is true:

\begin{lemma}\label{lemtech}
  Let $(c_u)_{u \in \Zln}$ and $(d_u)_{u \in \Zln}$ be two geometric
  points of $V_J$ giving the same associated group  $\{ P_i, i \in \Zl \}$. Then there exist $\psi \in \mathfrak{H}$ such that $(d_u)_{u \in \Zln}= \psi.(c_u)_{u \in \Zln}$.
\end{lemma}
\begin{proof}
  First, up to an action of type~\eqref{eq_action1}, we can suppose that for
  all $i \in \Zl$, we have $P^{(c_u)_{u \in \Zln}}_{i} = P^{(d_u)_{u \in \Zln}}_i$. Thus there exist $\lambda_i \in
  \kbar$ such that $(c_{\rho(j,i)})_{j \in \Zn} = \lambda_i (d_{\rho(j,i)})_{j \in \Zn}$.
  Since $(c_u)_{u \in \Zln}$ and $(d_u)_{u \in \Zln}$ are projective, we can
  assume that $\lambda_0=1$.
  We will show that up to an action of type~\eqref{eq_action2}, for every $i
  \in \Zl$ such that $P_i$ is well defined, $\lambda_i=1$. But first we show
  that for such points, we have $\lambda_i^\ell=1$. 

  Let $i\in \Zl$ be such that $(c_{\rho(j,i)})_{j \in
    \Zn}$ is a well defined projective point. 
  Let $x,y,u,v \in \Zstruct{2\overn}$ which are
  congruent modulo $\Zn$, we remark that for $\mu \in \{1,
  \ldots, \ell \}$, $\rho(x, \mu.i)$, $\rho(y,i)$, $\rho(u,0)$,
  $\rho(v,0)$ are elements of $\Zstruct{2\overln}$ congruent modulo
  $\Zln$. Calling Theorem~\ref{thetarel}, we obtain
  that
\begin{equation}\label{eq21}
\begin{split}
  \big(\sum_{t \in Z(2)} \chi(t) c_{\rho(x+y+t,(\mu+1).i)} c_{\rho(x-y+t,(\mu-1).i)}\big).\big(\sum_{t \in Z(2)} \chi(t)
  c_{\rho(u+v+t,0)} c_{\rho(u-v+t,0)}\big)= \\
  =\big(\sum_{t \in Z(2)} \chi(t) c_{\rho(x+u+t,\mu.i)} c_{\rho(x-u+t,\mu.i)}\big).\big(\sum_{t \in Z(2)} \chi(t)
  c_{\rho(y+v+t,i)} c_{\rho(y-v+t,i)}\big), 
\end{split}
\end{equation}
for any $\chi \in \hat{Z}(2)$.

We have a similar formula involving $(d_u)_{u \in \Zln}$. Using
equation (\ref{eq21}) and an easy recurrence, we obtain that $\lambda_{\mu.
i}=\lambda_i^{u_{\mu}}$ where $(u_{\mu})$ is a sequence such that $u_0=0$, $u_1=1$ and
$u_{\mu+1}+u_{\mu-1}=2.u_\mu+2$.
The general term of this sequence is $u_\mu=\mu^2$. For $\mu=\ell$, we
have 
\begin{equation}\label{lemtecheq1}
\lambda_i^{\ell^2}= \lambda_{\ell.i} = \lambda_0 = 1
\end{equation}

Now, by the symmetry relations, we have for $j \in \Zn$,
$c_{\rho(j,\mu.i)}=c_{\rho(-j,-\mu.i)}$. Applying this for $\mu=1$ and
$j=0$, we obtain that $\lambda_i=\lambda_i^{(\ell-1)^2}$ which together with
(\ref{lemtecheq1}) gives
\begin{equation}
\lambda_i^\ell =1
\end{equation}
which concludes the claim.

Let $(e_1,\dots,e_g)$ be the canonical basis of $\Zl$. Up to an action of
type~\eqref{eq_action2} we may assume that $\lambda_{e_i}=1$ and
$\lambda_{e_i + e_j}=1$ for $i,j \in \{1,\dots,g\}, j < i$. Now let $a,b
\in \Zl$ be such that $\lambda_a=1$, $\lambda_b=1$ and $\lambda_{a-b}=1$. 
Then by Theorem~\ref{thetarel} we have the relations:
\begin{equation}\label{eq21bis}
\begin{split}
  \big(\sum_{t \in Z(2)} \chi(t) c_{\rho(x+y+t,a+b)} c_{\rho(x-y+t,a-b)}\big).\big(\sum_{t \in Z(2)} \chi(t)
  c_{\rho(u+v+t,0)} c_{\rho(u-v+t,0)}\big)= \\
  =\big(\sum_{t \in Z(2)} \chi(t) c_{\rho(x+u+t,-b)} c_{\rho(x-u+t,b)}\big).\big(\sum_{t \in Z(2)} \chi(t)
  c_{\rho(y+v+t,a)} c_{\rho(y-v+t,a)}\big). 
\end{split}
\end{equation}
Since by symmetry, $\lambda_{-b}=1$, the relations~\eqref{eq21bis} give that
$\lambda_{a+b}=1$. An easy recurrence shows that for any $i \in \Zl$ we have
$\lambda_i=1$, which concludes the proof.
\end{proof}

As a first application of this lemma, we have:
\begin{proposition}\label{prop:reduced}
  If $\ell$ is prime to the characteristic of $k$ and $\nu \geq 2$
  then $V_J$ is a reduced scheme.
\end{proposition}
\begin{proof}
  We recall that $V_J$ is the affine variety defined by $J$ where
  $J$ is the image of the homogeneous ideal $I$ defining $\Mbarln$,
  under the specialization map
  \begin{eqnarray*}
    k[x_u|u \in
    \Zln] \rightarrow k[x_u|u \in
    \Zln,n u \not=0], \quad x_u \mapsto \left\{
      \begin{array}{l@{,\hsp}l}
        {b_u} & \mathrm{if} \quad u \in \Zn \\
        {x_u} & \mathrm{else}
      \end{array} \right. ,
  \end{eqnarray*}
  with $(b_u)_{u \in \Zn}$ the theta null point associated to $(B_k,
  \ppol, \Thetan)$.

  By definition, $V_J$ is a closed subvariety of the affine space
  $\Aff^{\Zln}$. For $\lambda \in \Zl$, denote by $\pi_\lambda :
  \Aff^{\Zln} \rightarrow \Aff^{\Zn}$ the projection deduced from the
  inclusion $\phi_\lambda:k[x_u| u \in \Zn] \rightarrow k[x_u| u \in
  \Zln]$, $x_u \mapsto x_{\rho(u, \lambda)}$. In order to prove that
  $V_J$ is a reduced scheme it is enough to prove that for any $x$
  geometric point of $V_J$ and all $\lambda \in \Zl$, $\pi_\lambda (x)$
  is a reduced point of $\Aff^{\Zn}$. We consider two cases.
  
  If $\pi_\lambda(x)$ is not the point at origin of $\Aff^{\Zn}$ then
  it defines a projective point of $\proj^{\Zl}$ which is a
  $\ell$-torsion point of $V_{I_{\Thetan}}$ by Proposition
  \ref{prop_pi}. As a consequence, $\pi_\lambda(x)$ is contained in
  the reduced line $L$ between the origin point of $\Aff^{\Zn}$ and
  $\pi_\lambda(x)$. By the preceding lemma, the intersection of $V_J$
  with $L$ is contained in a variety isomorphic to $\spec{k[x]/(x^\ell
    -1)}$ where $x^\ell -1$ is a separated polynomial as $\ell$ is
  prime to the characteristic of $k$. We deduce that $x$ is a reduced
  point of $\Aff^{\Zn}$.
  
  If $\pi_\lambda(x)$ is the origin point of $\Aff^{\Zn}$, it is
  enough to prove that $\pi_\lambda(x)$ is reduced in the case that
  $\nu =2$. In fact, the set of equations generating $J$ in the case
  $\nu \geq 2$ contains the set of equations generating $J$ in
  the case $\nu=2$.  We suppose now that $\nu =2$. Let
  $\mathfrak{P}=(x_u | u \in \Zn)$ be the ideal of $k[x_u| u \in \Zn]$
  defining the reduced point at origin of $\Aff^{\Zn}$. Let $J_\lambda
  = J \cap \phi_\lambda (k[x_u| u \in \Zn])$ and denote by $J_{\lambda
    \mathfrak{P}}$ the local ring of $J_\lambda$ in $\mathfrak{P}$. As
  $J$ is a $0$-dimensional ideal, we know that there exist $m$ a
  positive integer such that $J_{\lambda \mathfrak{P}} \supset
  \mathfrak{P}^m$ in $k[x_u| u \in \Zn]_{\mathfrak{P}}$. Let
  $r_\lambda$ be the smallest integer with this property. We want to
  show that $r_\lambda = 1$. In order to do so, we are going to use
  another formulation of the Riemann relations given by Theorem
  \ref{riemannquad}.
  
  For this, we let $H(\overline{\ell n})=\Zln \times \dZtwo$ and
  $H(\overline{ n})=\Zn \times \dZtwo$.  We denote by $\rho':
  H(\overline{ n}) \times \Zl \rightarrow H(\overline{\ell n})$ the
  natural isomorphism deduced from $\rho$. For all $v=(v',v'') \in
  H(\overline{\ell n})$, we let $y_v = \sum_{t \in \Ztwo} v''(t)
  x_{v'+t}$. Let $a_1,a_2,a_3,a_4, \tau \in H(\overline{n})$ such that
  $2\tau = a_1-a_2-a_3-a_4$. Set $\alpha_1 = \rho'(a_1,2\lambda)$,
  $\alpha_2=\rho'(a_2,0)$, $\alpha_3=\rho'(a_3,0)$,
  $\alpha_4=\rho'(a_4,0)$ and $\tau_1=\rho'(\tau,\lambda)$ so that
  we have $2\tau_1 = \alpha_1-\alpha_2-\alpha_3-\alpha_4$. We write
  $\tau=(\tau',\tau'')$ and let $H(\overline{2})= \{ x\in
  H(\overline{\ell n}) | x\,\text{is}\, 2- \text{torsion modulo} \,
  \Ztwo \times \{ 0 \} \}$.  By applying \cite[formula (C'') p.
  334]{MR36:2621}, we have the following relation in $J$:
  \begin{equation}\label{eq:classical}
\begin{split}
  y_{\alpha_1}y_{\alpha_2} y_{\alpha_3}y_{\alpha_4} = & \\ = &
  \frac{1}{2^g} \sum_{t \in H(\overline{2})} (\tau'' + t'')(2t')
  y_{\alpha_1 - \tau_1 + t} y_{\alpha_2 + \tau_1 +t} y_{\alpha_3+
    \tau_1+t} y_{\alpha_4+ \tau_1 +t},
\end{split}
  \end{equation}
  where $t=(t',t'') \in H(\overline{2})$.

  By definition, for $i=2,3,4$, if we write $a_i=(a_i',a_i'')$, we
  have $y_{\alpha_i}=\sum_{t \in \Ztwo} a_i''(t) b_{a_i'+t}$. As by
  hypothesis $(b_u)_{u \in \Zn}$ is valid theta null points, by
  applying \cite[formulas (*) p. 339]{MR36:2621}, we obtain that for
  any $a_i=(a_i',a_i'') \in H(\overline{n})$ there exists $\beta_i'
  \in 2 \Zn$ such that $\sum_{t \in \Ztwo} a_i''(t)
  b_{a_i'+\beta_i'+t} \neq 0$. As a consequence, for any choice of
  $a_1$, we can find $a_2$, $a_3$, $a_4$, and $\tau \in
  H(\overline{n})$ such that $2\tau = a_1-a_2-a_3-a_4$ and for
  $i=2,3,4$, $y_{\alpha_i}= \sum_{t \in \Ztwo} a_i''(t) b_{a_i'+t}
  \neq 0$. (We can take for instance $a_1=a_2=a_3=a_4$ so that
  $a_1-a_2-a_3-a_4 \in 2 H(\overline{n})$ and then if necessary add to
  $a_2, a_3, a_4$ elements of $2 \Zn$ in order to have $y_{(a_i,0)}
  \neq 0$.) As an immediate consequence, we obtain that
  $\pi_{2\lambda}(x)$ is also the origin point of $\Aff^{\Zn}$.
  
  Let $r'_\lambda$ be the smallest integer such that $r'_\lambda \geq
  r_\lambda$ and $4 | r'_\lambda$. We remark that
  $\phi_{2\lambda}(k[x_u|u \in \Zn])=k[y_{\rho'(v,2\lambda)}| v \in
  H(\overline{n})]$. Let $M$ be a degree $r'_\lambda/4$ monomial in
  the variables $y_{\rho'(v,2\lambda)}$. If necessary, by multiplying
  $M$ by a suitable non null constant, we see that $M$ is equal to a
  product $M'$ of $r'_\lambda/4$ polynomials given by the right hand
  of~(\ref{eq:classical}). These polynomials have degree $4$ and are
  sums of products of monomials of the form $y_{\rho'(v,\lambda)}$
  (using the symmetry relations). We deduce from this that $M' \in
  \mathfrak{P}^{r_{\lambda}}$ and as a consequence $M' \in
  J_{\lambda}$. But this means that $M \in J_{2\lambda}$ and as $M$
  can be any degree $r'_\lambda/4$ monomial in the variables
  $y_{\rho'(v,2\lambda)}$, we have proved that $J_{2\lambda
    \mathfrak{P}} \supset \mathfrak{P}^{r'_\lambda/4}$.

  Let $m$ be an integer such that $2^m \lambda = \lambda$ in $\Zl$.
  Using the previous result and an easy recurrence, we see that if
  $r_\lambda > 1$ then $r_\lambda = r_{2^m \lambda} < r_\lambda$ which
  is a contradiction.
\end{proof}

As a second application of Lemma~\ref{lemtech}, we have:
\begin{theorem}\label{main}
  Let $(a_u)_{u \in \Zln}$ be a geometric point of $V_J$.
  For any $i \in \Zl$, let $P_i$ be the geometric point,
  if well defined, of $\proj_k^{\Zn}$ with homogeneous
  coordinates $(a_{\rho(j,i)})_{j \in \Zn}$. Denote by $S$
  the subset of $\Zl$ such that $P_i$ is a well defined
  projective point for all $i \in S$. If $K=\{ P_i , i \in S \}$ is a
  maximal $\ell$-torsion subgroup of
  $V_{I_{\Thetan}}$ of rank $g$ then
  $(a_u)_{u \in \Zln}$ is a well defined theta null
  point. In other words there exists $(A_k, \pol,
  \Thetaln)$ an abelian variety together with a
  $(\overln)$-marking with associated theta null point
  $(a_u)_{u \in \Zln}$.
\end{theorem}
\begin{proof}
  Let $A_k$ be the quotient of $B_k \simeq
  V_{I_{\Thetan}}$ by $K$ and let $\pi: B_k \rightarrow
  A_k$ be the canonical isogeny.  As $K$ is a subgroup of 
  $B_k[\ell]$, there exists an isogeny $\hat{\pi}:
  A_k \rightarrow B_k$ such that $[\ell]=\hat{\pi} \circ \pi$. Let
  $\pol = \hat{\pi}^*(\ppol)$. We are going to show that there exists
  a certain theta structure $\Thetaln$ such that the
  theta null point associated to $(A_k, \pol, \Thetaln)$ is  
  $(a_u)_{u \in \Zln}$.

  Let $K(\ppol)=K_1(\ppol) \times K_2(\ppol)$ be the decomposition into
  isotropic subspaces for the commutator pairing induced by the theta
  structure $\Thetan$.  Denote by $\hat{K}$ the kernel
  of $\hat{\pi}$. As $\pol=\hat{\pi}^*(\ppol)$, we
  know that $\hat{K}$ is an isotropic subgroup of $K(\pol)$ for the
  commutator pairing. Moreover, by construction it is contained in the
  $\ell$-torsion subgroup of $A$ and by hypothesis has rank $g$.  We
  choose a decomposition as isotropic subspaces
  $K(\pol)=K_1(\pol) \times K_2(\pol)$ such that $\hat{K}$ is contained
  in $K_2(\pol)$ and for $i=1,2$, $\pi(K_i(\pol)[n])=K_i(\ppol)$.

  Denote by $\kappa : G(\pol) \rightarrow K(\pol)$ the natural
  projection.  By the descent theory of Grothendieck, there exists a
  unique level subgroup $\tilde{K}_\ell$ of $G(\pol)$ contained in
  $\kappa^{-1}(\hat{K})$ such that the quotient of $(A_k,\pol)$ by the
  action defined by $\tilde{K}_\ell$ gives $(B_k,\ppol)$.  Let
  $G^*(\pol)$ be the centralizer of $\tilde{K}_\ell$ in $G(\pol)$. By
  \cite[Prop. 2 pp. 291]{MR34:4269}, we have an isomorphism
  $$i:G^*(\pol)/\tilde{K}_\ell \simeq G(\ppol).$$
  Let $G(\pol)[n] = \kappa^{-1}(K(\pol)[n])$. We remark that
\begin{enumerate}
\item $G(\pol)[n]$ is contained in $G^*(\pol)$,
\item $\kappa(G(\pol)[n]\cap \tilde{K}_\ell)$ is the zero subgroup of $A_k$.
\end{enumerate}
Let $\tilde{K}_0$ be the level subgroup of $G(\ppol)$ defined as the
image by $\Thetan$ of the subgroup $(1,0,y)_{y \in \dZn}$ of $\Hn$.
An immediate consequence of 1. and 2. is that there exists a unique
level subgroup $\tilde{K}_n$ of $G(\pol)$ such that
$i(\tilde{K}_n)=\tilde{K}_0$.

Denote by $\tilde{K}_2$ the level subgroup of $G(\pol)$ whose
restriction over $K(\pol)[n]$ and $K(\pol)[\ell]$ is respectively
given by $\tilde{K}_n$ and $\tilde{K}_\ell$. By construction, we have
\begin{equation}\label{eq10}
i(\tilde{K}_2)= \tilde{K}_0.
\end{equation}

Choose any theta structure $\Thetaln:
\Hln \rightarrow G(\pol)$ such that the image
by $\Thetaln$ of the subgroup $(1,0,y)_{y \in
  \dZln}$ is exactly $\tilde{K}_2$. Because of
(\ref{eq10}) and construction of Proposition \ref{prop1}, we have
$\vartheta_0^{\Thetaln}=
\hat{\pi}^*(\vartheta_0^{\Thetan})$.

We suppose moreover that $\Thetaln$ is such that for
all $x\in \Zn$, $i(\Thetaln(1,x,0))=\Thetan(1,x,0)$, where we consider
$\Zn$ as a subgroup of $\Zln$ via the map
$x \mapsto \ell x$.  We remark that by construction,
$\Thetaln$ and $\Thetan$ verify the
conditions (\ref{eq1}) and (\ref{eq2}) and as a consequence are
$\hat{\pi}$-compatible. As a consequence of Corollary \ref{coro1}, we
have that for all $i \in \Zn$,
$\vartheta_i^{\Thetaln}=
\hat{\pi}^*(\vartheta_i^{\Thetan})$.

Let $(a'_u)_{u \in \Zln}$ be the theta null point associated to $(A_k,
\pol, \Thetaln)$. For $i \in \Zn$, denote by $Q_i$ the geometric point
of $\proj_k^{\Zn}$ with homogeneous coordinates $(a'_{\rho(j,i)})_{j
  \in \Zn}$.

We know that the projective coordinates of a maximal isotropic
$\ell$-torsion subgroup of $A_k$ is obtained by the action of the
theta group on $(a'_u)_{u \in \Zln}$ by translation.  Denote by $K'$
the $\ell$-torsion subgroup of $A_k$ given by the points with
projective coordinates $(a'_{v+i})_{v \in \Zln}$. By construction, $K'$
is the dual of $\hat{K}$ for the commutator pairing which implies that
$A_k$ is exactly the quotient of $B_k$ by $\hat{\pi}(K')$. As a
consequence, we have $\hat{\pi}(K')=K$.

The applications $\Zl \rightarrow B_k[\ell]$, $j \mapsto P_j$ is a
group morphism (see for instance the proof of \cite[Lemma 5.6]{pc3}),
as well as the application $\Zl \rightarrow B_k[\ell]$, $j \mapsto
\hat{\pi}(Q_j)$. By changing the theta structure $\Thetaln$, we can
suppose that for all $j \in \Zl$, $\hat{\pi}(Q_j)=P_j$. As a
consequence, for $j \in \Zl$ there exists $\lambda_j \in \overline{k}$
such that for $i\in \Zn$, $a_{\rho(j,i)} = \lambda_j a'_{\rho(j,i)}$.
We know moreover that $(a_u)_{u \in \Zln}$ and $(a'_u)_{u \in \Zln}$
are geometric points of $V_J$. Applying Lemma \ref{lemtech} we are
done.
\end{proof}


If $(a_u)_{u \in \Zln}$ is a geometric point of $V_J$, we denote by
$G((a_u)_{u \in \Zln})$ the subgroup of $B_k[\ell]$ generated by the
valid projective points $(a_{\rho(j,i)})_{j \in \Zn}$ for $i \in \Zl$
of $V_{I_{\Thetan}}=B_k$. The preceding theorem tells us that whenever
a solution $(a_u)_{u \in \Zln}$ of $J$ is such that $G((a_u)_{u \in
  \Zln})$ is a maximal $\ell$-torsion subgroup of $B_k[\ell]$ then it
is a valid theta null point, that is, it corresponds to a certain
$(A_k, \pol, \Thetaln)$. It would be desirable to be able to determine
which maximal rank $g$ subgroups of $B_k[\ell]$ can arise as a $G(x)$ where
$x$ is a geometric point of $V_J$ representing a valid theta null
point.

For this, let $\bpol = [\ell]^* \ppol$ on $B_k$. As $\ppol$ is
symmetric, we have that $\bpol \simeq \ppol^{\ell^2}$ and as a
consequence $K(\bpol)$, the kernel of $\bpol$ is isomorphic to
$Z(\overline{\ell^2 n})$. The polarisation $\bpol$ induces a
commutator pairing $e_{\bpol}$ on $K(\bpol)$ and as $\bpol$ descend to
$\ppol$ via the isogeny $[\ell]$, we know that $e_{\bpol}$ is trivial
on $B_k[\ell]$. For $x_1,x_2 \in B_k[\ell]$, let $x'_1, x'_2 \in B_k[\ell^2]$ 
be such that $\ell.x'_i =x_i$ for $i=1,2$. We remark that $x'_1$ and $x'_2$
are defined up to an element of $B_k[\ell]$. As a consequence, $e_{\bpol}(x'_1,x_2)
=e_{\bpol}(x_1,x'_2)$, does not depend on the choice of $x'_1$ and $x'_2$ 
and if we put $e_W(x_1,x_2)=e_{\bpol}(x'_1,x_2)$, we obtain a well defined 
bilinear application $e_W : B_k[\ell] \times B_k[\ell] \rightarrow
\overline{k}$. As $e_{\bpol}$ is a perfect pairing, for any $x'_1 \in
B_k[\ell^2]$ there exits $x'_2 \in B_k[\ell^2]$ such that
$e_{\bpol}(x'_1, x'_2)$ is a primitive $\ell^{2th}$ root of unity. As a
consequence, for any $x_1 \in B_k[\ell]$ there exists $x_2 \in
B_k[\ell]$ such that $e_W(x_1, x_2)$ is a primitive $\ell^{th}$ root
of unity and $e_W$ is also a perfect pairing.

\begin{theorem}
  A maximal $\ell$-torsion subgroup of $B_k$ of rank $g$ is of the form
  $G(x)$ where $x$ is a geometric point of $V_J$ corresponding to a
  valid theta null point if and only if $G(x)$ is an isotropic
  subgroup for the pairing $e_W$.
  \label{th:gpisotropes}
\end{theorem}
\begin{proof}
  Let $(a_u)_{u \in \Zln}$ be a geometric point of $V_J$ corresponding
  to a valid theta null point.  We know that $(a_u)_{u \in \Zln}$ is
  the theta null point of a triple $(A_k, \pol, \Thetaln)$. The theta
  structure $\Thetaln$ induces a decomposition $K(\pol)=K_1(\pol)
  \times K_2(\pol)$ into isotropic subgroups for the commutator
  pairing $e_{\pol}$. As the isogeny $\pi$ is such that
  $\pi^*(\vartheta_i^{\Thetan})=\vartheta_i^{\Thetaln}$ for all $i \in
  \Zn$ (and identifying $i \in \Zn$ with $\ell i \in \Zln$), we know
  that $G((a_u)_{u \in \Zln})= \pi(K_1(\pol))$. We denote by
  $\hat{\pi} : B_k \rightarrow A_k$ the isogeny such that $\pi \circ
  \hat{\pi}=[\ell]$ as in the diagram (\ref{eq_diag1}). For any $x_1, x_2
  \in G((a_u)_{u \in \Zln})$, there exists $\overline{x}_1,
  \overline{x}_2 \in K_1(\pol)[\ell]$ such that $x_i=\pi(\overline{x}_i)$, $i=1,2$. Let $x'_1 \in B_k[\ell^2]$ be such that
  $\ell.x'_1=x_1$. We have
  $e_W(x_1,x_2)=e_{\bpol}(x'_1,x_2)=e_{\pol}(\hat{\pi}(x'_1),
  \hat{\pi}(x_2))$. But $\hat{\pi}(x_2)=\hat{\pi}\circ
  \pi(\overline{x}_2)=[\ell](\overline{x}_2)=0$. As a consequence, we have $e_W(x_1,x_2)=0$.
 
  Now, we prove the opposite direction. Let $G$ be a maximal rank
  $g$ $\ell$-torsion subgroup of $B_k[\ell]$ which is
  isotropic for the pairing $e_W$ and $\hat{G}$ be the dual group of
  $G$ for the pairing $e_W$. As $e_W$ is a perfect pairing, $\hat{G}$
  is also a maximal rank $g$ $\ell$-torsion subgroup of $B_k[\ell]$.
  We want to show that $G$ is of the form $G(x)$ with $x$ a geometric
  point of $V_J$ where $J$ is defined by the triple $(B_k, \ppol,
  \Thetan)$.  For this, we consider the isogeny $\hat{\pi}:B_k
  \rightarrow A_k$ with kernel the subgroup $G$ of $B_k$. As $G$ is
  contained in $B_k[\ell]$, $G$ is an
  isotropic subgroup of $(B_k, \bpol)$, and $\bpol$ descend via
  $\hat{\pi}$ to a polarization $\pol$ on $A_k$. Let $\pi : A_k
  \rightarrow B_k$ be the isogeny with kernel $\hat{\pi}(\hat{G})$.
  By the commutativity of the following diagram,
\begin{equation}
  \xymatrix{
    (B_k, \bpol) \ar[dd]^{[\ell]} \ar[dr]^{\hat{\pi}} &  \\
    & (A_k,\pol) \ar[dl]^{\pi}   \\
    (B_k,\ppol) & \\
  },
  \label{eq_diag2}
\end{equation}
$\pol$ descends via $\pi$ to $\ppol$.

The theta structure $\Thetan$ induces a decomposition
$K(\ppol)=K_1(\ppol) \times K_2(\ppol)$. Let $x_i=\hat{\pi}(x'_i)$
with $x'_i \in \hat{G}$ and $i=1,2$. Let $y'_1 \in B_k[\ell^2]$ be such
that $\ell.y'_1 =x'_1$. We have by hypothesis $1=e_W(x'_1, x'_2)=e_{\bpol}(y'_1,
x'_2)$ and as a consequence $1=e_{\bpol}(x'_1,
x'_2)=e_{\pol}(x_1,x_2)$. Thus
$\hat{\pi}(\hat{G})$ is isotropic for the pairing $e_{\pol}$. As a
consequence, we can chose a decomposition $K(\pol)=K_1(\pol)\times
K_2(\pol)$ such that for $i=1,2$, $\pi(K_i(\pol))=K_i(\pol_0)$ and
$K_2(\pol)[\ell]=\hat{\pi}(\hat{G})$. Take any theta structure
$\Thetaln$ for $\pol$ compatible with this decomposition. Let
$(a_u)_{u \in \Zln}$ be the associated theta null point. By 
Corollary~\ref{coro1}, $(a_u)_{u \in \Zln}$ is a geometric point of
$V_J$. Moreover, we have $G((a_u)_{u \in \Zln})=\pi(K_1(\pol))=G$.

\end{proof}

Our study of valid theta null points allows us to better understand
the geometry of $V^0_J$. We know from Proposition~\ref{prop_geometrie}
that $V^0_J$ classifies the isogenies $\pi: A_k \to B_k$ between
marked abelian varieties verifying the compatibility condition.

Taking the dual of $\pi$ gives an isogeny from $B_k$ to $A_k$ with
kernel $K=\pi(K_1)=\{P_1,\ldots,P_n \}$. Thus the theta null points on
$V_J^0$ correspond to varieties $\overline{\ell}$-isogeneous to $B_k$.
But we have seen in Proposition~\ref{lem_damien} that it may happen
that different points of $V_J^0$ give the same kernel $K$ and hence
the same isogeneous variety. We want to classify the points of $V_J^0$
corresponding to isomorphic varieties $\overline{\ell}$-isogeneous to $B_k$.

To do that, let $K$ be a maximal isotropic subgroup of rank $g$ of the
points of $\ell$-torsion of $B_k$. We are interested in the class
$\mathfrak{T}_K$ of isogenies of kernel $K$.  More precisely, if $\pi$
is an isogeny from $B_k$ to $A_k$ with kernel $K$, then
$\mathfrak{T}_{K}$ is the class of isogenies $\pi': B_k \to A'_k$ such
that there exists an isomorphism $\psi: A_k \to A'_k$ that makes the
following diagram commutative:
\begin{equation*}
\xymatrix{
&          &                    & A_k \ar[dd]^{\psi}\\
0 \ar[r] & K \ar[r] & B_k \ar[ur]^{\pi} \ar[dr]^{\pi'}&  \\
         &          &                    & A'_k 
}
  \label{eq_isogenyiso}
\end{equation*}

\begin{proposition}
  \label{prop_structure}
  Let $K$ be a maximal subgroup of rank $g$ of the points of
  $\ell$-torsion of $B_k$ which is isotropic for the pairing $e_W$.
  There is a point $(a_u)_{u \in \Zln} \in V^0_J$ such that the
  corresponding dual isogeny $\pi: B_k \to A_k$ is in
  $\mathfrak{T}_K$. Every other point in $V^0_J$ giving the class
  $\mathfrak{T}_K$ is obtained by the action of $\mathfrak{H}$ on
  $(a_u)_{u \in \Zln}$. In particular, the geometric points of
  $V^0_J/\mathfrak{H}$ are in bijection with the
  $\overline{\ell}^g$-isogenies of $B$.
\end{proposition}
\begin{proof}
  Let $K=\{P_i, i \in \Zl\}$ be such a maximal subgroup.
  Theorem~\ref{main} gives a geometric point $(a_u)_{u \in \Zln}$ of
  $V^0_J$ corresponding to a marked abelian variety $(A_k, \pol_A,
  \Theta_A)$ such that the associated isogeny $\tilde{\pi}: A_k \to
  B_k$ sends $K_1(\pol_A)$ to $K$. Hence, the unique isogeny $\pi: B_k
  \rightarrow A_k$ such that $\tilde{\pi} \circ \pi = [\ell]$, is in
  $\mathfrak{T}_K$. If $(a'_u)_{u \in \Zln}$ is another valid theta
  null point in $V^0_J$, corresponding to a marked abelian variety
  $(A', \pol_{A'}, \Theta_{A'})$ such that the dual of the associated
  isogeny gives the same class as $\pi$, then we have the following
  diagram:
\begin{equation*}
\xymatrix{
    & A_k \ar[dl]^{\tilde{\pi}} \\
B_k &  \\
    & A'_k \ar[ul]^{\tilde{\pi}'} \ar[uu]^{\tilde{\psi}}
}
  \label{eq_isogenytildeiso}
\end{equation*}
By definition of the associated isogenies $\tilde{\pi}$ and $\tilde{\pi}'$, we
know that $\pol_A=\tilde{\pi}^{*}(\pol_B)$ and
$\pol_{A'}=\tilde{\pi}'^{*}(\pol_B)=\tilde{\psi}^{*}(\pol_A)$. So $\tilde{\psi}$
induces a morphism of the theta groups $G(\pol_A)$ and $G(\pol_{A'})$, and
pulling back by the theta structures we get a symmetric automorphism
$\tilde{\psi}$ of $\mathcal{H}(\overln)$. Since the theta structures
$\Theta_A$ and $\Theta_{A'}$ are compatible with $\Theta_B$, $\tilde{\psi}$ is
in $\mathfrak{H}$. This shows that $(a_u)_{u \in \Zln}$ and $(a'_u)_{u \in
\Zln}$ are in the same orbit under $\mathfrak{H}$.
\end{proof}



Together with the study of degenerate theta null points, it is now
possible to count the points in $V_J$. For instance, take $g=1$, $n=4$
and $\ell =3$. Let $E$ be an elliptic curve, and $(b_u)_{u \in
  (\Z/n\Z)}$ be a level $4$ theta null point on $E$.  There are
$4=\#\proj^1(\F_3)$ classes of $3$-isogenies from $E$, and $6=3 \times
\phi(3)$ solutions in $V_J$ for each class. The
actions~\eqref{eq_action1} are given by $(a_u)_{u \in (\Z / \ell n
  \Z)}  \mapsto (a_{x .u })_{u \in \Z / \ell n \Z}$ where $x \in \Z /
\ell n \Z$ is invertible and congruent to $1$ mod~$n$. There are
$\phi(\ell)$ such actions. The actions~\eqref{eq_action2} are given by
$(a_u)_{u \in \Z / \ell n \Z}  \mapsto (\zeta^{c.u^2} a_{u })_{u \in
  \Z / \ell n \Z}$ where $\zeta$ is a $\ell^{th}$-root of unity and $c
\in \Z/ \ell \Z$.

 If $g=2$, it is easy to compute the number of valid theta null point
 in $V_J$. 
 First, we remark that the number of isogeny classes of
 degree $\ell^2$ of a given dimension $2$ abelian variety $B_k$ is
 parametrised by the points of a Grassmanian $Gr(2,4)(\F_\ell)$ which are
 isotropic (see Theorem~\ref{th:gpisotropes}):
 there are $(\ell^2+1)(\ell+1)$ such points.

 Next, the number
 of actions of the form (\ref{eq_action1}) is parametrised by the
 number of invertible matrices of dimension~$2$ with coefficients in
 $\F_\ell$ with is given by $(\ell^2-1).(\ell^2-\ell)$. The number of
 actions of the form (\ref{eq_action2}) is $\ell^3$ (the number of symmetric
 matrices of dimension~$2$). As a consequence,
 the number of valid theta null point in $V_J$ is
 \[\ell^{10}-\ell^8-\ell^6+\ell^4.\]
We remark that this number is a $O(\ell^{11})$. For $g=2$, $\ell=3$, we
have $51840$ valid theta null points in $V_J$.

For a general $g$ and $\ell$, we assess the order of the number of
valid theta null point which are solution of $V_J$. The number of
isotropic points of a Grassmanian $Gr(g,2.g)(\F_\ell)$ is a 
$O(\ell^{g(g+1)/2})$. 
The
number of action of the form (\ref{eq_action1}) is a $O(\ell^{g^2})$
and the number of action of the form is a $O(\ell^{g(g+1)/2})$. We deduce that
the number of valid theta null point in $V_J$ is bounded by 
\begin{equation}\label{eq:nbsolutions}
O(\ell^{2.g^2+g}).
\end{equation}

\begin{example}
  In the case of genus $1$ and small $\ell$ it is possible to list
  all the solutions of $V_J$. We take $\ell=3$ and let $E$ be the
  elliptic curve given by an affine equation $y^2 = x^3 + 11.x + 47$
  over $\F_{79}$. A corresponding theta null point of level~$4$ for
  $E$ is $(1 : 1 : 12 : 1)$.  The four subgroups of $3$-torsion of $E$
  are
\begin{align*}
K_1 &=\{(1: 1: 12 : 1), (37 : 54 : 46 : 1), (8 : 60 : 74 : 1) \} \\
K_2 &=\{(1: 1: 12 : 1), (67 : 10 : 68 : 1), (62 : 8 : 70 : 1) \} \\
K_3 &=\{(1: 1: 12 : 1), (42 : 5 : 15 : 1),  (40 : 16 : 3 : 1) \} \\
K_4 &=\{(1: 1: 12 : 1), (72 : 56 : 31 : 1), (69 : 24 : 33 : 1) \} \\
\end{align*}
All geometric points of $V_J$ are defined over $\F_{79}(\upsilon)$ where
$\upsilon$ is a root of the irreducible polynomial $X^3 + 9.X + 76$.
For each of the four subgroups $K_i$, there are $6$ geometric points of 
$V_J$ giving the curve $E/K_i$. We give a point in each class (the other
points can be obtained via the actions~\eqref{eq_action1}
and~\eqref{eq_action2}): 

$Q_1=(16\upsilon^2 + 19\upsilon + 17: 1: 46: 16\upsilon^2 + 19\upsilon + 17: 37: 54: 34\upsilon^2 + 70\upsilon + 46: 54: 37: 16\upsilon^2 + 19\upsilon + 17: 46: 1)$ corresponds to $K_1$.

$Q_2= (64\upsilon^2 + 67\upsilon + 68: 1: 68: 64\upsilon^2 + 67\upsilon + 68: 67: 10: 57\upsilon^2 + 14\upsilon + 26: 10: 67: 64\upsilon^2 + 67\upsilon + 68: 68: 1) $ corresponds to $K_2$.

$Q_3= (8\upsilon^2 + 49\upsilon + 48: 1: 3: 8\upsilon^2 + 49\upsilon + 48: 40: 16: 17\upsilon^2 + 35\upsilon + 23: 16: 40: 8\upsilon^2 + 49\upsilon + 48: 3: 1)$ corresponds to $K_3$.

$Q_4= (32\upsilon^2 + 73\upsilon + 34: 1: 33: 32\upsilon^2 + 73\upsilon + 34: 69: 24: 68\upsilon^2 + 7\upsilon + 13: 24: 69: 32\upsilon^2 + 73\upsilon + 34: 33: 1)$ corresponds to $K_4$.

We also have the following degenerate points in $V_J$: if we take
$x=9$ in the action~\eqref{eq_action1}, the image of the class of any
$Q_i$ is $\mathcal{C}= \{ (55 : 1 : 12 : 55 : 1 : 1 : 28 : 1 : 1 : 55
: 12 : 1), (1 : 1 : 12 : 1 : 1 : 1 : 12 : 1 : 1 : 1 : 12 : 1), (23 : 1
: 12 : 23 : 1 : 1 : 39 : 1 : 1 : 23 : 12 : 1) \} $. For this class, the
corresponding $\ell$-torsion subgroup (the points $P_i$ of
Proposition~\ref{prop_pi}) is $ \{ (1 : 1 : 12 : 1), (1 : 1 : 12 : 1),
(1 : 1 : 12 : 1)\}$ which has rank~$0$. On $\mathcal{C}$ the
action~\eqref{eq_action1} is trivial, so there are only $3$ points in
this degenerate class, coming from the action~\eqref{eq_action2}. The
last degenerate point is $ (1 : 0 : 0 : 1 : 0 : 0 : 12 : 0 : 0 : 1 : 0
: 0)$, alone in its class.
\end{example}

We conclude this section with some remarks concerning the case $\nu
=1$ and the case where the characteristic of $k$ is equal to $\ell$.
First, for computational reasons, for instance in order to limit the
number of variables when computing the points of $V_J$, we would like
to have $\nu$ as small as possible.  All the results of Section
\ref{s_solutions} are valid under the hypothesis that $\nu \geq 2$ and
that the characteristic of $k$ is different from $\ell$. In the case
$\nu =1$, we can not even prove that $V_J$ is a zero dimensional
variety. Nonetheless we have made extensive computations which back
the idea that even in the case $\nu =1$, in general, $V_J$ is a zero
dimensional variety whose degree is of the same order with respect to
the parameter $\ell$ as in the case $\nu =2$.

In the case that the characteristic of the base field $k$ is equal to
$\ell$ and $\nu \geq 2$, the proof that $V_J$ is a $0$-dimensional
scheme is still valid. In this case $V_J$ is not anymore reduced and
the computation of the number of solutions of $V_J$ are not valid.
Nonetheless, from our computations, we see that in this case the
degree of the variety $V_J$ is of the same order with respect to the
parameter $\ell$ as in the case where the characteristic of $k$ is
different from $\ell$.

In the following section, we give an algorithm to find the solutions
of $V_J$. We can prove that this algorithm is efficient in the case
$\nu \geq 2$ and when the characteristic of $k$ is different from
$\ell$. In the case that $\nu =1$ or when the characteristic of $k$ is
equal to $\ell$ we will make the hypothesis that $V_J$ is a zero
dimensional variety whose degree  is given by formula
(\ref{eq:nbsolutions}). Under these hypothesis, we can also prove that our
algorithm is efficient.

\section{An efficient algorithm} \label{sec:FastAlgo}
We would like to use the formulas of Section \ref{s_modcorr} to
compute the image of the modular correspondence $\Phi_\ell$ for some
positive integer $\ell$. We have seen that the main algorithmic
difficulty is to solve the polynomial system defined by the equations
of Theorem \ref{riemannquad} together with the symmetry relations. The
aim of this section is to give an algorithm to solve efficiently this
system.  We have made an implementation of our algorithm and used it
to test the heuristics described at the end of Section
\ref{s_solutions}.

Let $n=2^\nu$. In this section, $k$ is a finite field. We let $(B_k,
\ppol, \Thetan)$ be a dimension $g$ abelian variety together with a
$\overn$-marking and we denote by $(b_u)_{u \in \Zn}$ its associated
theta null point. Let $J$ be the image of the homogeneous ideal
defining $\Mbarln$ given by the equation of Theorem \ref{riemannquad},
under the specialization map
\begin{eqnarray*}
  k[x_u|u \in
  \Zln] \rightarrow k[x_u|u \in
  \Zln,n u \not=0], \quad x_u \mapsto \left\{
    \begin{array}{l@{,\hsp}l}
      {b_u} & \mathrm{if} \quad u \in \Zn \\
      {x_u} & \mathrm{else}
    \end{array} \right. .
\end{eqnarray*}
We denote by $V_J$ the $0$-dimensional affine variety (heuristically
$0$-dimensional if $\nu=1$) defined by the ideal $J$. Let $\rho : \Zn
\times \Zl \rightarrow \Zln$ be the group isomorphism given by $(x,y)
\mapsto \ell x+ n y$

\subsection{Motivation}\label{subsec:motivation}
In order to find the points of the variety $V_J$ a first idea is to
use an efficient Gr\"obner basis computation algorithm \cite{MR1213453}
such as F$_4$ \cite{MR1700538}. We have carried out computations in
the case $g=2$, $\nu =1$ and $\ell=3$ with respect to a total degree
order (the DRL \cite{MR1287608,CLO92} or grevlex order) using
the computer algebra system Magma \cite{magma} implementation of
F$_4$. From our computation, we could conclude that
\begin{itemize}
\item even for a small coefficient field (\(k=\mathbb{F}_{3^{10}}\)),
  it takes 20 hours of computations using Magma on a powerful computer
  with 16 Go of RAM;
\item as expected from the computations of Section \ref{s_solutions},
  the number of solutions in the algebraic closure $\overline{k}$
  of $k$ is big: \(30853\) solutions in characteristic \(3\)
  (We note that this is coherent
  with the number of solutions discussed after
  Proposition~\ref{prop_structure}
  when $g=2$, $\nu=2$ and $\ell=3$).
\item to fully solve the system (that is to say, find explicitly all
  the solutions in $\overline{k}$) we need to compute a second
  Gr\"obner basis with respect to a lexicographical order.
\end{itemize}
This last operation can be done using the FGLM \cite{MR1263871}
algorithm. In our case it is equivalent to compute the characteristic
polynomial of a \(30853\times30853\) matrix. This computation did not
finish using Magma for the base field $k=\mathbb{F}_{3^{10}}$. So we
see that even for $g=2$, $\nu=1$ and $\ell=3$ the computation of the points of
$V_J$ is painful using a generic algorithm.  In this section, we give
an algorithm to solve efficiently the algebraic system defined by $J$
for small $\ell$ over a big coefficient field. As an application of
our method, we can mention the initialisation phase of a point
counting algorithm \cite{pc3}.

The main idea of our algorithm is to use explicitly the symmetry
inside the problem deduced from the action of the theta group: we
compute a Gr\"obner basis not for the whole ideal $J$ but rather a
Gr\"obner basis of a well chosen projection \(J\cap k[x_{\rho(v,\lambda)} |
v \in \Zn ]\) for $\lambda \in \Zl$. With our strategy, the same problem
(\(k=\mathbb{F}_{3^{10}}\)) can be solved in seconds and far bigger
problems (\(k=\mathbb{F}_{3^{1500}}\)) can be solved in less than $1$
hour (see Section \ref{experimental results} for experimental
results).
\subsection{Assumptions}\label{strategy}

Our method is a combination of existing algorithms.  We first describe
in full generality the assumptions upon which our algorithm is faster
than a general purpose Gr\"obner basis algorithm.  Then, using the
results of Section \ref{s_solutions}, we explain that these
assumptions hold for $J$ in the case that $\nu \geq 2$ and that the
characteristic of $k$ is not $\ell$. If $\nu=1$ or if the
characteristic of $k$ is equal to $\ell$, we can not prove the
assumptions but we have made extensive computations which show that in
general our algorithm is much more efficient than a general purpose
Gr\"obner basis algorithm.

Let $T$ be a set $[x_1, \ldots , x_s]$ of variables, we assume that
$J\subset k[T] $ is a zero dimensional ideal generated by the
polynomials $[f_1,\ldots,f_m]$ where for $i=1, \ldots ,m$, $f_i$ is a
polynomial in $k[T]$.  We make the hypothesis that we can split the
set of variables into two subsets $T = X \cup Y$ such that the ideal
$K=J\cap k[Y]$ contains low degree polynomials.

In order to make precise what we mean by low degree polynomials, we
denote by $I_{gen}$ an ideal generated by the polynomials
$[g_1,\ldots,g_m]$ where for $i=1, \ldots, m$, $g_i$ is a general
polynomial of total degree $\deg(f_i)$. We define for any ideal $I$ of
$k[T]$:
$$
D_{Y}(I)=\min\{deg(g)\mid 0\neq g\in I\cap k[Y]\}.
$$
Our assumption that $J\cap k[Y] $ contains low degree polynomials
means that
\begin{equation}\label{hyp11}
D_{Y}(J) \ll D_{Y}(I_{gen}).\tag{H1-1}
\end{equation}
The previous assumption implies that our algorithm will perform much 
faster with the particular ideal $J$ than it would do for a general
ideal $I_{gen}$.

We must also ensure that it is more efficient to compute a Gr\"obner
basis for \(J\cap k[Y]\) instead of a Gr\"obner basis for \(J\). If we
suppose that a Gr\"obner basis computation for a total degree order
has the same complexity for $J$ and $I_{gen}$, we have to check that
\(D_{Y}(J) \ll D_{T}(I_{gen}).\) It is well known that, generically, a lower bound
for \(D_{T}(I_{gen})\) is given by the Macaulay bound which is given
by \(D_{T}(I_{gen})=1+\sum_{i=1}^m(\deg(f_i)-1)\) if $m \leq s$.  We
can now state explicitly the second part of our first assumption:
\begin{equation}
  D_{Y}(J) \ll \sum_{i=1}^m\deg(f_i).
  \tag{H1-2}
\label{hyp12}
\end{equation}
Our second assumption is that $J$ can be decomposed into many prime
ideals. There exists a positive integer $r\gg 1$ such that
\begin{equation}
\sqrt{J}=P_{1}\cap\cdots\cap P_{r} \text{ and }P_{i}%
\text{ is a prime ideal}\tag{H2}.
\label{hyp2}
\end{equation}
We recall that for a homogeneous ideal we define the Hilbert function
$\operatorname{HF}_{I}(d)=\dim(k[T]/I)_{d}$ and the degree of the
ideal $I$, $\deg(I)$, is given by the Hilbert series
$\sum\limits_{i=0}^{\infty}\operatorname{HF}_{I}(d)\,z^{i}=\frac
{M(z)}{(1-z)^{\dim(I)}}$ and $\deg(I)=M(1)\neq0$. With this, we can
state the third (optional) assumption
\begin{equation}
\deg\left(  \sqrt{I}\right)  \ll\deg\left(  I\right).\tag{H3}
\label{hyp3}
\end{equation}

We discuss the validity of hypothesis (\ref{hyp11}), (\ref{hyp12}),
(\ref{hyp2}) and (\ref{hyp3}) in the case that $J$ is defined as in
the introduction of the present section.  First, we remark that
$D_{Y}(I_{gen})$ can be easily computed: let $M(s,d)$ be the number of
monomials of degree less or equal to $d$ in $s$ variables.  The total
number of solutions counted with multiplicities of $I_{gen}$ is given
by the B\'ezout bound: $D=\Pi_{i=1}^m \deg(f_i)$. Hence, we have
\begin{equation}\label{eq:bound}
D_{Y}(I_{gen})=\min_d \left\{ M(h,d)>D \right\},
\end{equation}
where $h$ is the cardinal of $Y$ and $M(h,d)={h+d \choose d}$. By
considering $M(h,d)$ as a polynomial in the unknown $d$, we obtain
that for a given $h$, $D_Y(I_{gen})$ is the biggest real root of the
polynomial:
$$\frac{1}{h!}\prod_{i=1}^h (x+i)=D.$$  
As a consequence, we have
\begin{equation}\label{eq:eval1}
D_{Y}(I_{gen}) \sim_{D \rightarrow\infty} (h! D)^\frac{1}{h}.
\end{equation}

We know moreover that $\Mbarln$ has dimension $1/2.g.(g+1)$ and is
embedded via the relations given in Theorem \ref{riemannquad} in the projective
space of dimension $(n\ell)^g-1$. We deduce that $J$ contains at least
$(n\ell)^g-1/2.g.(g+1)$ algebraically independent polynomials. As the
equations of Theorem \ref{riemannquad} have degree $4$, a lower bound
for $D$ is $4^{(n\ell)^g-1/2.g.(g+1)-1}$.

On the other side, if we chose for $j \in \Zl$, $Y=[x_{\rho(u,j)}| u
\in \Zn]$, we know by Proposition \ref{prop_pi} that the solutions of
the system $J \cap k[Y]$ can be either the origin point of
$\Aff^{\Zn}$ or represent a $\ell$-torsion point of $V_{I_{\Thetan}}$.
In this last case, by Lemma \ref{lemtech} we know that there is $\ell$
solutions of $J$ corresponding to the same projective points. Denote
by $D'$ the number of solutions of $J\cap k[Y]$ counted with
multiplicities. We have  $D' \leq \ell^{2g+1}+1$ and using the
heuristic evaluation of $D_Y(J)$ given by (\ref{eq:bound}), we obtain
\begin{equation}\label{eq:eval2}
D_Y(J)\sim_{D \rightarrow\infty} (h! D')^\frac{1}{h}.
\end{equation}
For a fixed $g$ and $\nu$ the cardinal of $Y$ are fixed.
Using (\ref{eq:eval1}) and (\ref{eq:eval2}), we see that hypothesis
(\ref{hyp11}) is verified for $\ell$ big enough.

Next, $1+\sum_{i=1}^m(\deg(f_i)-1)=3.(n\ell)^g$. On the other side,
$(h! D')^\frac{1}{h}$ with $D'\leq \ell^{2g+1}+1$ and $h=n^g$. As
$n\geq 2$, we have, using the Stirling approximation formula, that $(h!
D')^{\frac{1}{h}} =O(\ell)$ and hypothesis (\ref{hyp12}) is
verified as soon as $g \geq 2$ and $\ell$ big enough.

Since we want to find at least one solution of $J$ defined over $k$, we
can assume that such a solution exists. By Proposition \ref{prop_pi},
this implies that there exists a subgroup $G$ of rank at least $1$ of
the $\ell$-torsion group of $V_{I_{\Thetan}}$ such that all the points
of $G$ are defined over $k$. As the solutions of $J\cap k[Y]$ are
points of $V_{I_{\Thetan}}[\ell]$, we conclude that for $r \geq \ell$
we have:
$$\sqrt{J}=P_{1}\cap\cdots\cap P_{r} \text{ and }P_{i}%
\text{ is a prime ideal},$$
and hypothesis (\ref{hyp2}) is verified.

In general, we know from Proposition \ref{prop:reduced} that the
hypothesis (\ref{hyp3}) is not verified since $J$ is a reduced ideal.
Nonetheless, in the case that the characteristic of $k$ is
equal to $\ell$, the scheme defined by $J$ is not reduced and we use
(\ref{hyp3}) in order to speed up the computations.

\subsection{General strategy}\label{subsec:general}
In the following, we give a general strategy for computing the
solutions of the algebraic system defined by $J$.  All the steps of
our algorithm are standard with the exception of step $1$ and step
$4$. In step $1$, we try to use as much as possible the assumptions
(\ref{hyp11}) and (\ref{hyp12}) and step $4$ is based upon the
assumptions (\ref{hyp2}),(\ref{hyp3}).

\begin{enumerate}[Step 1]
\item Using a specific algorithm given in Section \ref{description},
  we compute a truncated Gr\"obner basis for an elimination order
  and a modified graduation. This allows us to obtain an zero
  dimensional ideal $J_{1}$ contained in $J$. In general $J_1$ is not
  equal to $J$. The output of the algorithm is a sequence of
  polynomials $[p_{1},\ldots,p_{\kappa}]$ in $k\left[ Y\right] $ such that
  $J_1$ is generated by $\left( p_{1},\ldots,p_{\kappa}\right)$.

\item Compute a Gr\"obner basis $G_{\text{DRL}}$ of $J_{1}$ for a total
  degree order (DRL or grevlex). This can be done with any
  efficient algorithm for computing Gr\"obner basis, for instance
  F$_{4}$.

\item Compute a Gr\"obner basis $G_{\text{Lex}}$ of $J_{1}$ for a
  lexicographical order. This can be done by using the FGLM
  algorithm to change the monomial order of $G_{\text{DRL}}$.

\item Compute a decomposition into primes of the following ideal:
\[
\sqrt{J_{1}}=P_{1}\cap\cdots\cap P_{r}%
\]

We assume that $\deg(P_{i})=1$ (if it is not the case we replace
$k$ by some algebraic extension of $k$).\newline

\item For $i$ from $1$ to $r$, we repeat the following Steps a,b,c for the
ideal $(P_{i})+I$:

\begin{enumerate}
\item Compute a Gr\"obner basis $G_{i}$ of $(P_{i})+I$ for a total
  degree order (DRL).

\item Change the monomial order to obtain $G_{i}^{\prime}$ a
  lexicographical Gr\"obner basis of $(P_{i})+I$ .

\item Compute a decomposition into primes: $\sqrt{P_{i}+I}=P_{j_{i-1}+1}%
\cap\cdots\cap P_{j_{i}}$ (by convention $j_{-1}=r$).
\end{enumerate}
\end{enumerate}

Since we have $\sqrt{I}=\sqrt{J_{1}\cap I}=\sqrt{P_{1}\cap
I}\cap\cdots\cap \sqrt{P_{r}\cap I}$ and since the decomposition of each
component $\sqrt{P_{i}\cap I}$ is done by step 5 of the previous
algorithm, we obtain a decomposition of the ideal $I$:
\[
\sqrt{I}=P_{r+1}\cap\cdots\cap P_{j_{r}}%
\]

\begin{remark}
Once we have obtained a point $P$ of $V_J$ corresponding to a valid
theta null point, we can recover all the solutions of $V_J$
corresponding to valid theta null points using the action given by Proposition
\ref{lem_damien}.
\end{remark}

\subsection{Description of the algorithm}\label{description}
In this section, we give a detailed explanation of the Step 1 and Step
4 of the algorithm described in Section \ref{subsec:general}.

\textbf{Step 1: elimination algorithm}

The normal strategy for computing Gr\"obner bases (Buchberger, $F_{4}$,
$F_{5}$) consists in considering first the pairs with the minimal
total degree among the list of critical pairs (see
\cite{CLO92,BeWe93}, for instance).

In the following, to select critical pairs, we consider only the total
degree with respect to the first set of variables $X$. More precisely:

\begin{definition}
Partial degree of critical pair $p=(f,g)$:%
\begin{multline*}
\deg_{X}\left(  p\right)  =\text{total degree of }\operatorname{lcm}\left(
\operatorname*{LT}_{<}(f),\operatorname*{LT}_{<}(g)\right)  \\ \text{ in the
polynomial ring }R\left[  X\right]  \text{ where }R=k\left[
Y\right].
\end{multline*}

\end{definition}

Moreover, we stop the computation of the Gr\"obner basis as soon as we find a
zero dimensional system in $k\left[  Y\right]  $. 
Consequently we obtain an new version of the $F_{4}$ algorithm:

\begin{algorithm}
Algorithm $F_{4}$ (modified version)\newline\linebreak
\begin{center}
\frame{%
\begin{tabular}
[c]{l}%
Input: $\ \left\{
\begin{array}
[c]{l}%
F\ \text{a finite subset of}\ k[x_{1},\ldots,x_{s}]\\
<\text{ a monomial admissible order}\\
X=\left[  x_{1},\ldots,x_{\kappa}\right]  \text{ and }Y=\left[  x_{\kappa+1}%
,\ldots,x_{s}\right]
\end{array}
\right.  $\\
\textbf{Output}: a finite subset of $k[x_{1},\ldots,x_{s}]$%
.\newline\\
$G:=F$ and $P:=\left\{  \operatorname{CritPair}(f,g)\ |\ (f,g)\in
G^{2}\ \text{with }\ f\neq g\right\}  $\newline\\
\textbf{while} $P\neq\emptyset$ and $\dim(G\cap k\left[  Y\right]
)>0$ \textbf{do}\newline\\
$\ \ \ d:=$ $\min\left\{  \deg_{X}\left(  p\right)  \mid p\in P\right\}  $
minimal partial degree of critical pairs\\
$\ \ $extract from $P,$ $P_{d}$ the list of critical pairs of degree $d$\\
$\ \ \ R:=$\textsc{Matrix\_Reduction}$(\operatorname{Left}(P_{d}%
)\cup\operatorname{Right}(P_{d}),G)$\\
\ \ \textbf{for } $h\in R$ \textbf{do}\newline\\
$\ \ \ \ \ \ P:=P\cup\left\{  \operatorname{CritPair}(h,g)\ |\ g\in G\right\}
$\newline\\
\ \ \ \ \ \ $G:=G\cup\left\{  h\right\}  $\newline\\
\textbf{return} $G$\newline%
\end{tabular}
}
\end{center}
\end{algorithm}

\textbf{Step 4: decomposition into primes}

\label{decomposition} The known general purpose algorithms to compute
a primary decomposition of an ideal are inefficient in our case.  To
speed up the computation, we proceed following the three steps:

\begin{enumerate}
\item[Step 1] The basis $G_{\text{Lex}}$ always contains a univariate polynomial
$g(x_{s})$. We can factorize this polynomial. We will see that this is the
most consuming part of the whole algorithm. We obtain
\[
g(x_{s})=f_{1}(x_{s})^{\alpha_{1}}\cdots f_{l}(x_{s})^{\alpha_{l}}%
.\]

\item[Step 2] For all factors $i$ from $1$ to $l$ we apply the lextriangular
algorithm \cite{Laza92}
to obtain efficiently a decomposition into triangular sets of
$J_{1}+\left\langle f_{i}\left(  x_{s}\right)  \right\rangle $.

We can describe the algorithm beginning by the special case of two variables
$\left[  x_{s-1},x_{s}\right]  $ (this enough in our case since we assume that
$k=$ $\overline{k}$ as we will see later). By a theorem of
Lazard \cite[Theorem 1]{Laza85},
the general shape of $G_{\text{Lex}}$ the lexicographical order Gr\"obner
basis is as follows:%
\begin{equation}
\left\{
\begin{array}
[c]{l}%
g\left(  x_{s}\right) \\
h_{1}\left(  x_{s-1},x_{s}\right)  =g_{1}\left(  x_{s}\right)  \left(
x_{s-1}^{k_{1}}+\cdots\right) \\
h_{2}\left(  x_{s-1},x_{s}\right)  =g_{2}\left(  x_{s}\right)  \left(
x_{s-1}^{k_{2}}+\cdots\right) \\
\cdots\\
h_{s}\left(  x_{s-1},x_{s}\right)  =x_{s-1}^{k_{s}}+\cdots\\
\text{polynomials in variables }x_{1},\ldots,x_{s}%
\end{array}
\right. \label{lextriangular}%
\end{equation}

with $k_{1}<k_{2}<\cdots<k_{s}$ and $g_{1}\left(  x_{s}\right)  \mid
g_{2}\left(  x_{s}\right)  \mid\cdots$. Hence we can obtain for free some
factors of $g(x_{s})$:

\item[Step 3]
\begin{align*}
g(x_{s})  & =\left(  \frac{g(x_{s})}{g_{1}(x_{s})}\right)  \,g_{1}(x_{s})\\
& =\left(  \frac{g(x_{s})}{g_{1}(x_{s})}\right)  \,\left(  \frac{g_{1}(x_{s}%
)}{g_{2}(x_{s})}\right)  g_{2}(x_{s})\\
& =\cdots
\end{align*}

For any factor $f_{i}\left(  x_{s}\right)  $ of $g(x_{s})=f_{1}(x_{s}%
)^{\alpha_{1}}\cdots f_{l}(x_{s})^{\alpha_{l}}$, it is enough to find the
first element $h_{j}\left(  x_{s-1},x_{s}\right)  $ of the Gr\"obner basis such
that%
\[
\gcd\left(  f_{i}\left(  x_{s}\right)  ,g_{j}\left(  x_{s}\right)  \right)
\neq0.
\]

In our case $k=\overline{k}$ and each factor is linear $f_{i}\left(
  x_{s}\right) =x_{s}-\beta_{i}$ so that we search for the first $j$
such that $g_{j}\left( \beta_{i}\right) \neq0$: then we obtain a new
polynomial in one variable $h_{j}\left( x_{s-1},\beta_{i}\right) $
that can be factorized. Hence we can iterate the algorithm for all the
other variables $x_{s-2},\ldots,x_{1}$.
\end{enumerate}

\subsection{First experiments and optimizations}
In this section, we give running times for an implementation of the
strategy that we have presented in Section~\ref{strategy}. We also
explain some important optimizations.

The main motivation of the examples presented is this section, is to
show that the initialisation phase of the point counting algorithm
described in~\cite{pc3} can be made efficient enough to be negligible
in the overall running time of the algorithm. For this, we take $g=2$
and $n=2$ and we work over a field $k$ of characteristic $3$ or $5$.
We construct a theta null point of level~$2$ corresponding to an
abelian variety $A_k$ of dimension~$2$. We construct the modular
correspondence of level $\ell$ where $\ell$ is the characteristic of
$k$. Any valid solution of the modular correspondence will corresponds
to the theta null point of level~$2\ell$ of an abelian variety
isogeneous to $A_k$. We can then use the algorithm of~\cite{pc3} to
count the number of points of $A_k$.

\subsubsection{First experiments}

As explained in~\ref{subsec:motivation} if we can try to compute directly a Gr\"obner basis of the ideal generated by the
equations, even when $k$ is very small ($k=\mathbb{F}%
_{3^{10}}$ for instance), it takes 10 hours of computations on a powerful computer with 16
Go of RAM just to compute a DRL\ Gr\"obner basis. Moreover, in characteristic
$3$, there is a huge number of solutions: $30853$. This imply that there is no
hope to solve efficiently the corresponding problem directly.

Keeping the notations of the beginning of Section \ref{sec:FastAlgo},
we apply the method described in \ref{subsec:general} to find the
solutions of $J$. We let $\nu=1$, $\ell=3$ and $g=2$ so that $\Zln=(\Z
/ 6 \Z)^2$. Let $T=[x_u | u \in \Zln]$. For $j \in \Zl$, we define
$Y=[x_{\rho(u,j)}|u \in \Zn]$. Taking $j=\rho(0,1)$ and in the
following, for $u=(i,j) \in \Zln$, we let $x_u=x_{ij}$. With these
notations, we take $Y=\left[ x_{31},x_{32},x_{02},x_{01}\right]$ and
$X=T-Y$ the set of all other variables. Then we consider $J$ embedded
in the polynomial ring $k[T]$ where $k$ is $\mathbb{F}_{3^{k}}$ or
$\mathbb{F}_{5^{k}}$. In that case $J\cap k\left[ x_{31}%
  ,x_{32},x_{02},x_{01}\right] =J\cap k\left[ Y\right] $ is an ideal
of degree $160$ (to be compared with $30853$ the degree of the whole
ideal $J$). When $k=\mathbb{F}_{5^{k}}$ (resp. $k%
=\mathbb{F}_{3^{k}}$) the polynomial $g(x_{s})$ obtained in section
\ref{decomposition} is a square-free polynomial of degree $124$ (resp. a
non square-free polynomial of degree $70$). We report in the following
table some first experiments using the algorithm of section
\ref{subsec:general} implemented in Magma and in C (see section
\ref{experimental results} for a full description of the experimental
framework). First we consider only very small example:

\begin{center}%
{\small
\begin{tabular}
[c]{l|llll}%
Algo \ref{subsec:general} & Step 1 & Step 2 + Step 3 & Step 4 & Step 5\\\hline
$k=\mathbb{F}_{5^{10}}$ & 0.35 sec & 0.25 sec & 3.24+0.01=3.25 sec &
8.0+0.77+0.01+0.08=8.86 sec\\
$k=\mathbb{F}_{5^{20}}$ & 0.35 sec & 1.14 sec & 28.4+0.04=28.44 sec &
39.3+9.1+0.05+0.49=48.94 sec
\end{tabular}
}
\end{center}

Even if the theoretical complexity is linear in the size of $k$ it is
clear that, in practice, the behavior of the algorithm is not linear
in $\log(k)$. Moreover, when we increase the size of $k$, step 5
becomes the most consuming part of our algorithm. Hence, even if the
new algorithm is efficient enough to solve the problem for a small
base field $k$, the problems become intractable when $\#k>5^{100}$. In
the next paragraph we propose several optimizations to overcome this
limitation.

\subsubsection{Optimizations}

The idea is to apply \emph{recursively} the algorithm of section
\ref{subsec:general} to perform the step 5: we split again the first of variable
into two parts:\ $X=X^{\prime}\cup Y^{\prime}=X^{\prime}\cup\left[
x_{42},x_{21},x_{51},x_{12}\right]  $.

\begin{center}%
\begin{tabular}
[c]{l|ll}%
Algo \ref{subsec:general} & Original Step 5 & Recursive Step 5\\\hline
$k=\mathbb{F}_{5^{10}}$ & 8.0+0.77+0.01+0.08=8.86 sec &
0.05+0.41+0.33+0.01=0.8 sec\\
$k=\mathbb{F}_{5^{20}}$ & 39.3+9.1+0.05+0.49=48.94 sec &
0.12+1.53+2.44+0.01+0.02=4.1 sec\\
$k=\mathbb{F}_{5^{40}}$ &  & 0.13+2.46+7.16+0.01+0.01=9.78 sec
\end{tabular}

\end{center}

When $k=\mathbb{F}_{3^{k}}$ we obtain in step 3 of the algorithm
\ref{subsec:general} the following lexicographical Gr\"obner basis:%

\[
\left\{
\begin{array}
[c]{l}%
g\left(  x_{01}\right)  \text{ of degree 70}\\
h_{1}\left(  x_{02},x_{01}\right)  =g_{1}\left(  x_{01}\right)  \left(
x_{02}^{2}+\cdots\right)  \text{ and }g_{1}\text{ of degree 39}\\
h_{2}\left(  x_{02},x_{0,1}\right)  =x_{02}^{3}+\cdots\\
\cdots\text{ polynomials in variables }x_{31},x_{32},x_{02},x_{01}%
\end{array}
\right.
\]

and thus we can split $g_{1}\left(  x_{01}\right)  $ into two factors:%
\begin{align*}
g_{1}\left(  x_{01}\right)   & =\left(  x_{01}+\alpha_{1}\right)
^{3}\,\left(  x_{01}+\alpha_{2}\right)  ^{9}\cdots\left(  x_{01}+\alpha
_{4}\right)  ^{9}\\
\frac{g\left(  x_{01}\right)  }{g_{1}\left(  x_{01}\right)  }  &
=x_{01}\,\left(  x_{01}+\beta_{1}\right)  ^{3}\,\left(  x_{01}+\beta
_{2}\right)  ^{9}\cdots\left(  x_{01}+\beta_{3}\right)  ^{9}%
\end{align*}

Hence the polynomial $g_{1}\left(  x_{01}\right)  $ can be efficiently
factorized when $k$ is big.

\subsection{Experimental results}

\label{experimental results}\textbf{Programming language -- Workstation}%
\newline The experimental results have been obtained with several Xeon
bi-processor 3.2 Ghz, with $16$ Gb of Ram. The instances of our
problem have been generated using the Magma software. We used the
Magma version $2.14$ for our computations. The F$_{5}$
\cite{MR2035234} algorithm have been implemented in C language in the
FGb software and we used this implementation for computing the first
Gr\"obner base. All the other computations are performed under Magma
including factorizing some univariate polynomials and computing
Gr\"obner bases using the F$_{4}$ algorithm.

\textbf{Table Notation}\newline The following notations are used in the tables
of Fig.1 \ and Fig.2 below:

\begin{itemize}
\item $k$ is the ground field, $k^{\prime}\supset k$
is the field extension. The practical behavior of our algorithm is strongly
depending on the size of $k^{\prime}$; hence, since $k$ is
fixed, the practical depends strongly on the degree of the field extension
$[k^{\prime}:k]$. In order to obtain consistent data in the
following tables we keep only the case $[k^{\prime}:k]=2$.

\item $T$ is the total CPU\ time (in seconds) for the whole algorithm.

\item $T_{\mathrm{Gen}}$ is the time for generating the Riemann
  equations and computing a valid level $2$ theta null point (Magma).

\item $T_{\mathrm{Grob}}$ is the sum of the Gr\"obner bases computations (FGb
and Magma).

\item $T_{\mathrm{Fact}}$ is the sum of the Factorization steps (Magma).

\item $T_{1}$ is the total time of the algorithm excluding generating the
equations: $T_{1}=T-T_{\mathrm{Gen}}$.
\end{itemize}

\begin{center}%
\begin{tabular}
[c]{cc|ccccc}%
$k$ & $k$' & $T_{\mathrm{Gen}}$ & $T_{\mathrm{Grob}}$ &
$T_{\mathrm{Fact}}$ & $T_{1}$ & T\\\hline
$5^{50}$ & $5^{100}$ & 1.9 & 2.7 & 9.3 & 12 & 14\\
$5^{70}$ & $5^{140}$ & 3.4 & 3.3 & 16.0 & 19 & 23\\
$5^{100}$ & $5^{200}$ & 19.5 & 15.9 & 116.7 & 133 & 152\\
$5^{150}$ & $5^{300}$ & 27.9 & 16.8 & 159.7 & 177 & 205\\
$5^{200}$ & $5^{400}$ & 141.3 & 57.3 & 401.0 & 459 & 600\\
$5^{250}$ & $5^{500}$ & 178.4 & 62.1 & 651.8 & 715 & 893\\
$5^{300}$ & $5^{600}$ & 227.8 & 86.7 & 935.3 & 1023 & 1251\\
$5^{350}$ & $5^{700}$ & 674.8 & 108.5 & 1306.1 & 1416 & 2091\\
$5^{400}$ & $5^{800}$ & 764.1 & 100.5 & 2411.3 & 2513 & 3277\\
$5^{450}$ & $5^{900}$ & 1144.0 & 165.3 & 2451.3 & 2619 & 3763\\
$5^{500}$ & $5^{1000}$ & 1070.1 & 185.4 & 2990.0 & 3177 & 4247\\
$5^{600}$ & $5^{1200}$ & 1979.5 & 273.5 & 4888.6 & 5164 & 7144\\
$5^{700}$ & $5^{1400}$ & 3278.0 & 422.5 & 6872.2 & 7297 & 10575\\\hline
\multicolumn{7}{c}{Fig 1: Algorithm $\ell=3$, characteristic of $k$ is $5$.
}%
\end{tabular}

\begin{tabular}
[c]{cc|ccccc}%
$k$ & $k$' & $T_{\mathrm{Gen}}$ & $T_{\mathrm{Grob}}$ &
$T_{\mathrm{Fact}}$ & $T_{1}$ & T\\\hline
$3^{80}$ & $3^{160}$ & 3.6 & 2.0 & 0.4 & 3 & 7\\
$3^{80}$ & $3^{160}$ & 3.6 & 2.0 & 0.2 & 3 & 6\\
$3^{200}$ & $3^{400}$ & 29.0 & 11.1 & 6.9 & 20 & 49\\
$3^{600}$ & $3^{1200}$ & 239.2 & 36.2 & 44.5 & 88 & 327\\
$3^{800}$ & $3^{1600}$ & 403.7 & 50.6 & 89.6 & 150 & 554\\
$3^{1000}$ & $3^{2000}$ & 591.8 & 61.8 & 151.0 & 225 & 816\\
$3^{1500}$ & $3^{3000}$ & 2122.0 & 137.7 & 474.5 & 666 & 2788\\
$3^{3000}$ & $3^{6000}$ & 11219.9 & 396.3 & 3229.6 & 3704 & 14923\\\hline
\multicolumn{7}{c}{Fig 2: Algorithm $\ell=3$, characteristic of $k$ is $3$.}%
\end{tabular}

\end{center}

\textbf{Interpretation of the results}

\begin{itemize}
\item In characteristic $3$, the hardest part is the generation of the
  equations and the computation of a valid level $2$ theta null
  point: $T_{\mathrm{Gen}}\approx T$. In characteristic, $5$ we have
  $T\approx3\,T_{\mathrm{Gen}}$.

\item The most consuming part in algorithm described in \ref{subsec:general} is the
univariate factorization. Moreover due to the implementation in Magma
$T_{\mathrm{Fact}}$ is not really linear in the size of $k$.

\item The algorithm is much more efficient in characteristic $3$ since:

\begin{itemize}
\item All the solutions occur with some multiplicity, hence we have to
  deal with non-square-free polynomials. As a consequence, the degree
  of the univariate polynomials can be decreased by taking the
  square-free part of the polynomials.

\item The corresponding Gr\"obner bases are in not in shape-position:
  as explain in section \ref{decomposition} we can split the
  univariate polynomial by taking a gcd.
\end{itemize}

\item The algorithm is very efficient since we can completely find the
  solutions of the ideal $J$ for sizes of the base field $k=3^{1500}$
  or $k=5^{700}$ which are interesting for point counting application.
\end{itemize}

\section{Conclusion}
In this paper, we have described an algorithm to compute modular
correspondences in the coordinate system provided by the theta null
points of abelian varieties together with a theta structure. As an
application, this algorithm can be used to speed up the initialisation
phase of a point counting algorithm \cite{pc3}. The main part of the
algorithm is the resolution of an algebraic system for which we have
designed a specific Gr\"obner basis algorithm.  Our algorithm takes
advantage of the structure of the algebraic system in order to speed
up the resolution. We remark that this special structure comes from
the action of the automorphisms of the theta group on the solutions of
the system which has a nice geometric interpretation. In particular we
were able count the solutions of the system and to identify which one
correspond to valid theta null points.

\bibliographystyle{alpha}
\bibliography{practical.bib}
\end{document}